%% file: sample-manuscript.tex
\documentclass[format=acmsmall,screen=true,nonacm]{acmart}

\AtBeginDocument{%
  \providecommand\BibTeX{{%
    \normalfont B\kern-0.5em{\scshape i\kern-0.25em b}\kern-0.8em\TeX}}}

\setcopyright{acmlicensed}
\copyrightyear{2025}
\acmYear{2025}
\acmDOI{XXXXXXX.XXXXXXX}

\acmJournal{TOIS}

\usepackage{multirow}
\usepackage{algorithm2e}
\RestyleAlgo{ruled}
\usepackage{color}
\usepackage{colortbl}
\usepackage{diagbox}
\newtheorem{theorem}{Proposition}
\usepackage{hyperref}
\usepackage{xcolor}
\newcommand{\positive}[1]{\colorbox{green!20}{$#1$}}
\newcommand{\negative}[1]{\colorbox{red!20}{$#1$}}
\newcommand{\checkValue}[1]{\IfSubStr{#1}{-}{\negative{#1}}{\positive{#1}}}
\usepackage{xstring}
\usepackage{graphicx}
\usepackage{wrapfig}

\begin{document}

\title[Learn to Rank Risky Investors]{Learn to Rank Risky Investors: A Case Study of Predicting Retail Traders’ Behaviour and Profitability}

\author{Weixian Waylon Li}
\authornotemark[1]
\email{waylon.li@ed.ac.uk}
\affiliation{%
  \institution{AIAI, School of Informatics, University of Edinburgh}
  \streetaddress{10 Crichton Street}
  \city{Edinburgh}
  \country{United Kingdom}
  \postcode{EH8 9AB}
}

\author{Tiejun Ma}
\authornotemark[1]
\email{tiejun.ma@ed.ac.uk}
\affiliation{%
  \institution{AIAI, School of Informatics, University of Edinburgh}
  \streetaddress{10 Crichton Street}
  \city{Edinburgh}
  \country{United Kingdom}
  \postcode{EH8 9AB}
}


\begin{abstract}
Identifying risky traders with high profits in financial markets is crucial for market makers, such as trading exchanges, to ensure effective risk management through real-time decisions on regulation compliance and hedging. However, capturing the complex and dynamic behaviours of individual traders poses significant challenges. Traditional classification and anomaly detection methods often establish a fixed risk boundary, failing to account for this complexity and dynamism. To tackle this issue, we propose a profit-aware risk ranker (PA-RiskRanker) that reframes the problem of identifying risky traders as a ranking task using Learning-to-Rank (LETOR) algorithms. Our approach features a Profit-Aware binary cross entropy (PA-BCE) loss function and a transformer-based ranker enhanced with a self-cross-trader attention pipeline. These components effectively integrate profit and loss (P\&L) considerations into the training process while capturing intra- and inter-trader relationships. Our research critically examines the limitations of existing deep learning-based LETOR algorithms in trading risk management, which often overlook the importance of P\&L in financial scenarios. By prioritising P\&L, our method improves risky trader identification, achieving an 8.4\% increase in F1 score compared to state-of-the-art (SOTA) ranking models like Rankformer. Additionally, it demonstrates a 10\%-17\% increase in average profit compared to all benchmark models.
  
\end{abstract}



\maketitle

\section{Introduction}
\label{sec:intro}

High-risk trading by individual traders, particularly in spread trading of Contracts for Difference (CFDs), significantly impacts financial markets through substantial trading volume and liquidity. 
In the United Kingdom, CFD trading accounts for approximately 10\% of the £1.2 trillion traded annually on the London Stock Exchange \cite{MA2022330}, contributing to liquidity and necessitating extensive hedging in main markets. 
This creates both opportunities and risks for market participants, making it critical for market makers to predict spread traders' behaviour to manage risks, maintain liquidity, and set competitive bid-ask spreads effectively \citep{dnn}.
However, predicting the behaviour and profitability of spread traders is challenging due to the complexity of the dynamic nature of markets, individual traders' heuristic behaviour and the diverse range of trading strategies used by traders~\cite{dnn}. 
In the realm of spread trading, traders aim to profit from the price differences in financial instruments. 
Their strategies are largely based on anticipated price movements, adding another layer of complexity to market behaviour \cite{brady_ramyar, MA2022330}.

Traditionally, identifying risky behaviours in financial transactions is viewed as either a prediction or anomaly detection task \cite{HILAL2022116429}. 
As demonstrated in Figure~\ref{fig:ranking-advantage}, both methods typically use absolute decision boundaries to make predictions.
Differing from the existing methods, we propose a shift in perspective, treating risky trader prediction as a ranking problem for individual traders. 
This is due to the fluid nature of risk thresholds, which can fluctuate over time. 
In this context, relative rankings is more meaningful than absolute thresholds, making ranking models particularly suitable. 
For example, according to the Bank of England’s Quarterly Bulletin\footnote{\url{https://www.bankofengland.co.uk/quarterly-bulletin/2021/2021-q3/no-economy-is-an-island-how-foreign-shocks-affect-uk-macrofinancial-stability}} for Q3 2021, the economic downturn triggered by the coronavirus pandemic had a severe impact on global asset markets, with major equity indices falling to nearly 30\% of their levels from early December 2019 by late March 2020. 
During this period,  even the high-risk traders might see reduced profits, leading standard classification algorithms to inaccurately categorise all traders as low-risk traders. 
Ranking algorithms, in contrast, can better identify the top risky traders in such periods.

\begin{figure}
    \centering
    \includegraphics[width=0.8\linewidth]{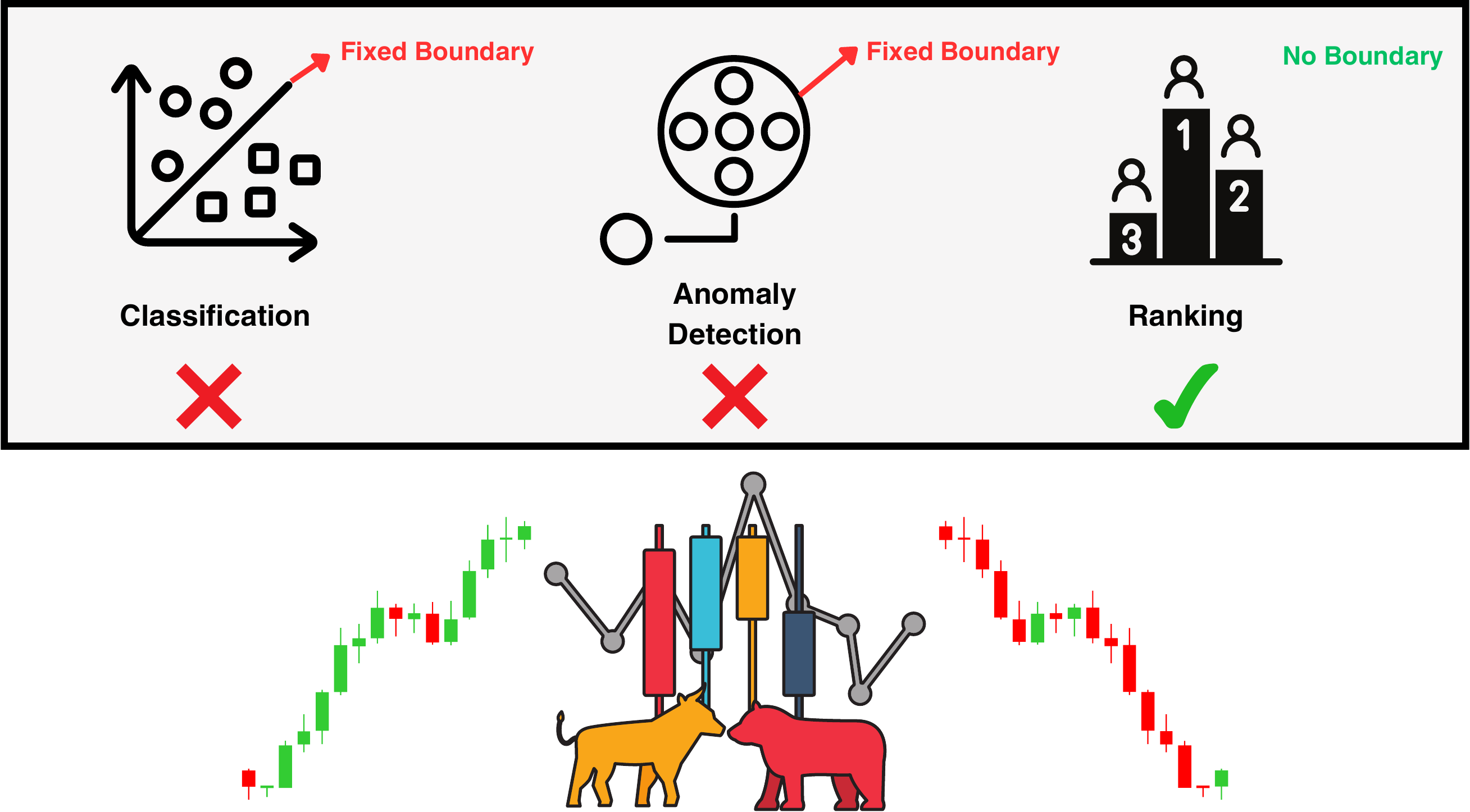}
    \caption{Comparison of classification, anomaly detection, and ranking approaches. Classification and anomaly detection rely on fixed decision boundaries, while ranking models use relative risk scores, better aligning with the dynamic market conditions.}
    \label{fig:ranking-advantage}
\end{figure}

However, adapting ranking algorithms to financial domains also presents unique challenges. 
Firstly, most existing ranking algorithms are designed for standard LEarning-TO-Rank (LETOR) tasks in web search \cite{10.1145/3292500.3330677,10.1145/3477495.3536334}, focusing on relevance rather than crucial financial metrics like risk adjusted return and profit and loss (P\&L). 
Additionally, previous research on LETOR models, tested on general information retrieval datasets like Istella LETOR\footnote{\begin{minipage}[t]{0.9\linewidth}\url{https://istella.ai/datasets/letor-dataset/}\end{minipage}}, 
MSLR30K\footnote{\begin{minipage}[t]{0.9\linewidth}\url{https://www.microsoft.com/en-us/research/project/mslr/?from=https://research.microsoft.com/en-us/projects/mslr/&type=exact}\end{minipage}}, and Yahoo! LETOR challenge \cite{pmlr-v14-chapelle11a}, do not adequately address the severe class imbalance found in risky trader detection.
Unlike typical information retrieval tasks, where each query usually has a number of relevant documents, financial datasets for risky trader detection are often highly imbalanced. 
This is characterised by a significant disproportion between the number of normal and anomalous instances, with some ranking groups in the finance context commonly lacking positive instances entirely. 
Lastly, deep learning based ranking models, though advanced, often lack the robustness of traditional gradient-boosted decision trees (GBDT) on several information retrieval benchmarks \cite{rankformer}. 
This limitation becomes pronounced when applied to financial data, as our experiments demonstrate a significant drop in performance within the financial domain.

To bridge the gap between deep learning based ranking algorithms and the unique requirements of the financial domain, this paper introduces profit-aware risk ranker (PA-RiskRanker), a novel LETOR method for ranking risky traders, particularly in spread trading scenarios with the challenge of market dynamics and highlight imbalanced datasets, and improve their effectiveness. Our main contributions are:

\begin{itemize}
    \item Introduction of a Profit-Aware binary cross entropy (PA-BCE) loss, which incorporates the profit information in the maximum likelihood function. This novel approach addresses data imbalance and integrates key financial indicators into the ranking model’s objective function. We successfully demonstrate that in financial risks where fixed decision boundaries are less effective, framing risk assessment as a ranking problem is more suitable than traditional classification or anomaly detection methods. This rethinking opens the door for future research, including applications in money laundering, fraud detection, and credit scoring.

    \item Development of PA-RiskRanker, an enhanced transformer-based ranking architecture, incorporating the Self-Cross Trader Attention pipeline and our PA-BCE loss. 
    PA-RiskRanker captures both the intra and inter relationship among risky traders, and boosts the ability to model financial behaviour for transformer-based ranking algorithms. 
    We validate this through comprehensive evaluations against SOTA ranking benchmarks, as well as classification and anomaly detection benchmarks. 
    PA-RiskRanker demonstrates a 10\%-17\% increase in average profit compared to all benchmark models. 
    Specifically, it outperforms the current SOTA deep learning ranker, Rankformer \citep{rankformer}, with approximately a 10\% improvement in F1 score and a 10.5\% increase in average profit.
    These results highlight PA-RiskRanker's superior capability in predicting risky traders.
\end{itemize}

\section{Background}

\subsection{Contracts For Difference (CFD)}

In this study, we mainly focus on the spread trading related to retail contracts for difference (CFDs) which are a type of financial instrument that allow individual traders and the market maker to enter a contract to exchange the difference in the asset's price between the time the contract is opened and closed \cite{CFDs}. 
Without actually owning the underlying asset, the retail traders are allowed to trade with leverage, which means they are only required to put down a small deposit of the entire trading value, increasing the potential profits from larger positions than the initial investment. 
While CFDs offer opportunities for leveraged trading, they introduce substantial risk to both traders and market makers, necessitating sophisticated risk management strategies.

In the United Kingdom, CFD trading is a crucial component of the financial market \cite{MA2022330}. 
Market makers in this sector provide liquidity and set prices, but they encounter unique challenges.
The substantial size of the CFD market exposes market makers to significant commercial viability risks.
They must meet strict collateral requirements and face the potential risk of failing to fulfill obligations to the CFD Counterparty within a rolling 12-month period, which could necessitate a parent company guarantee. 

Additionally, the profitability of market makers is closely tied to their ability to manage risks posed by traders, particularly those engaging in high-risk strategies or manipulation tactics such as insider trading and market manipulation \cite{cfdrisks}. 
Risky traders may illegally profit from activities like price manipulation, money laundering, and insider trading, placing market makers in unfavourable positions during market upswings \cite{HILAL2022116429, dnn}. 
These challenges underscore the importance of effectively managing the risks associated with high-risk traders in the CFD market. 
Therefore, market makers must make real-time risk management decisions regarding whether to hedge specific trades.

\subsection{Financial Data Classification and Anomaly Detection}

The detection of risky traders in tabular data, labelled as 1 (risky) or 0 (non-risky), is primarily a binary classification and anomaly detection problem. Several models have been developed to address these two types of problems. 

Support Vector Machines (SVMs) were applied as effective tools for detecting market manipulation, outperforming multivariate statistical techniques \citep{svm-ann}. Decision Trees (DT) have been used to create decision rules for market manipulation detection \citep{daiz-dt}. 
The emergence of GBDTs, such as XGBoost \cite{xgboost} and LightGBM \cite{NIPS2017_6449f44a}, marked a significant advancement, with these methods being applied to financial risk prediction and fraud detection \cite{ZHANG2023104045,xgboost-financial-fraud}. 
The advent of deep learning has further expanded its application to credit scoring \cite{10.1155/2020/8706285}. 
For instance, \cite{KWON2025125327} proposed a Siamese neural network to enhance both the predictive power and stability of credit scoring in dynamic environments.
However, recent studies indicate that deep learning methods often do not outperform traditional tree-based models in tabular data due to their weaker robustness \cite{NEURIPS2022_0378c769}. 
This has led to the development of tabular-specific deep learning architectures \cite{huang2020tabtransformer,NEURIPS2021_9d86d83f}. 
Though these architectures have shown performance comparable to traditional models, their effectiveness in financial risk assessment remains uncertain due to the lack of financial context evaluation.

In anomaly detection, the introduction of the Isolation Forest (IF) method presented a novel unsupervised approach to isolating anomalies \cite{4781136}. 
Its application in financial risk assessment, especially in credit card fraud detection, has been more recent \cite{TOKOVAROV2022433}. 
Subsequent developments like DevNet \cite{pang2019deep} and Deep SAD \cite{ruff2020deep} have shown impressive results using Gaussian prior-driven anomaly score distributions and one-class hyperspheres, respectively. Other techniques for tabular data anomaly detection, such as FeaWAD \cite{Zhou2021FeatureEW}, SLAD \cite{10.5555/3618408.3620019}, and Deep Isolation Forest (DIF) \cite{xu2023deep}, have also yielded superior outcomes. However, the robustness of these methods in financial risk detection tasks remains to be fully validated due to a lack of specific financial benchmarks, either. 

Both classification algorithms and anomaly detection methods aim to establish absolute decision boundaries for making predictions. 
However, a gap emerges when these are applied to financial markets, where the risk boundary is dynamic and comparative risk holds greater significance.

\subsection{Learning to Rank}

LETOR algorithms aim at ranking a list of items based on their features. It has been successfully applied in many areas such as recommendation systems~\cite{duan2010empirical,10.1145/3477495.3531931,rank-insights,10.1145/3312528}, question answering~\cite{yang2016beyond, reimers2019sentence}, retrieval~\cite{joachims2002optimizing, 10.1145/3477495.3531958,10.1145/3477495.3531948,10.1145/3439861}, etc. 
However, the application of LETOR in the finance area is limited. Most related research appeared in the recent years and they focus on only stock selection~\citep{zhang2022constructing,10.1145/3309547,10.1145/3690624.3709234}. Existing LETOR algorithms can be categorised into three groups based on their loss functions: pointwise, pairwise, and listwise.

Pointwise methods predict relevance scores for individual query-item pairs. 
Classic models like OPRF~\citep{fuhr1989optimum}, TreeBoost~\citep{friedman2001greedy}, and RankSVM~\citep{Shashua02taxonomyof} fall under this category, along with modern methods like RankCNN~\citep{severyn2015learning}. 
Although these models are straightforward and computationally efficient, they often oversimplify the ranking task by neglecting inter-item relationships. 
As a result, recent research in LETOR has moved away from pointwise models \cite{lee-etal-2024-methods}. 
Given that inter-trader relationships are crucial in our context, where comparing traders is essential, we will not consider pointwise LETOR methods in this study.

Pairwise methods, such as RankNet~\citep{burges2005learning}, LambdaRank~\citep{burges2006learning}, LambdaMART \cite{lambdarank} and FRank~\citep{tsai2007frank}, focus on the relative ranking of item pairs. They are more sophisticated than pointwise models as they consider inter-item relationships. However, they may not optimally capture the complex structure of financial data, where the overall ranking of a group of traders is more informative than pairwise comparisons.

Listwise approaches, including works leveraging advanced deep learning techniques like attention mechanisms~\cite{pang2020setrank,rankformer}, consider multiple items together. They are most suited for financial contexts as they directly maximise ranking metrics and can model complex relationships among a group of traders. These approaches are beneficial in scenarios where the relative ranking of traders is more informative than their individual scores, especially considering the fluid nature of financial markets.

Beyond classical methods, several lines of work have explored distillation-based learning to rank approaches to improve robustness, de-biasing, and computational efficiency in ad hoc retrieval and web search~\cite{10.1145/3681784,counterfactualletor,10.1109/TKDE.2022.3152585}.
Such methods are primarily motivated by challenges unique to document retrieval and typically rely on implicit feedback such as click data or relevance judgments.

In the financial domain, LETOR has been applied in stock selection, incorporating techniques like graph-based ranking and sentiment analysis~\citep{zhang2022constructing, song2017stock, poh2021building}. However, there remains a lack of LETOR models specifically designed for financial risk measurement, particularly for the detection of risky traders.

Given the intricate and interconnected nature of financial trading, a listwise approach is more suitable for modelling traders behaviour and the level of their trading sophistication. This approach can effectively capture the collective dynamics and the relative profitability skills of traders in a market, which is crucial for understanding and predicting risk in financial scenarios.

\section{Dataset}
\label{sec:dataset}

This study leverages a real-world exchange trading dataset encompassing transaction-level trading records from individual retail traders \cite{dnn}, which offers a rare opportunity to evaluate our proposed approach. Recognising the sensitivity and privacy considerations associated with individual-level financial data, this dataset is not made publicly available to ensure compliance with data protection regulations and ethical guidelines. 
Nonetheless, to foster transparency and facilitate reproducibility to the extent possible, we provide a comprehensive description of the dataset's structure, the nature of the data collected, and the rigorous statistical analysis conducted in this section.

The trading dataset comprises 13,607,120 real-life trading data records generated by 20,514 active traders from November 2003 to July 2014. Each record represents a single trade made by one of the active traders.

\subsection{Feature Creation}

The feature creation procedure aligns with prior research conducted on this dataset \cite{dnn}. This previous work underscored the inadequacy of relying solely on historical trading performance for accurate trader ability assessment, prompting the incorporation of additional trader-specific information and characteristics into the dataset.

The representativeness of the training data is achieved by considering three key aspects. Firstly, traders' behaviours are encoded based on their previous performance, which reflects their risk-adjusted gains and losses. Secondly, the time-ordered information about the traders is included. Finally, automatic feature extraction is performed to identify features related to traders' risk-adjusted profitability.

To address these aspects, three groups of features are constructed. The first group of features is derived from interviews with experienced members of the STX dealing desk, who assigned scores from 1 to 7 to different behavioural traits of the traders. These features focus on traders' demographics, such as age.

The second group of features captures information about traders' past performance. A rolling window of the previous 20 trades is considered, as determined by the hedging policy of stock exchange dealing desk. This group includes features such as average winning rate, average profit rate, average profit duration, and Sharpe ratio, which indicates the risk-adjusted return.

The third group of features provides information about traders' preferences for markets and channels. By examining the traders' previous 20 trades, the number of trades made in different market clusters is counted, allowing the identification of traders' favourite market clusters based on their trading history.

All the constructed features are detailed in Table~\ref{tab:feats_description}, which includes the ``Identifier'' and ``Future'' groups. The ``Future'' group contains features that provide future information as anticipatory characteristics. These features are only used for label creation and financial metrics calculation.

\begin{table}[ht]
\centering
\resizebox{\textwidth}{!}{%
\begin{tabular}{llll}
\toprule
Group & Feature & Type & Description \\ \midrule
\multirow{2}{*}{1} & Age1-5 & Discrete & One-hot encoding to indicate which age group the trader belongs to. \\
 & Mobile & Continuous & Proportion of trades that are made on mobile devices \\ \midrule
\multirow{14}{*}{2} & PerFTSE20 & Continuous & Share of trades placed in FTSE100 \\
 & AVGPTS3\_20 & Continuous & P\&L in Points \textgreater{}= 3 during the last 20 trades \\
 & ProfitRate20 & Continuous & Average profit rate of the trader in the past 20 trades \\
 & WinTradeRate20 & Continuous & Trader’s average winning rate in the past 20 trades \\
 & SharpeRatio20 & Continuous & Mean/st.dev. of returns in the past 20 trades \\
 & ProfitxDur20 & Continuous & Interaction of ProfitRate20 and DurationRate20 \\
 & PassAvgReturn & Continuous & Average return of the pass trades \\
 & AvgShortSales20 & Continuous & Share of short positions in the past 20 trades \\
 & DurationRate20 & Continuous & Average time trader leaves winning vs losing position open \\
 & AvgOpen20 & Continuous & Average of the P\&L among trader’s past 20 trades \\
 & DurationRatio20 & Continuous & Mean trade duration (mins) / std.dev. trade duration of the past 20 trades \\
 & OrderCloseRate20 & Continuous & \% of trades closed by an order in the past 20 trades \\
 & TradFQ20 & Continuous & The number of trades on average that a typical trader poses daily \\
 & NumTrades & Continuous & Trades accumulated until the last 20 trades \\ 
 \midrule
\multirow{2}{*}{3} & Segment1-3 & Discrete & Categories of past 20 trades' average return: \textgreater{}.05, between 0 and .05, negative \\
 & MarketCluster0-9 & Discrete & Identification of traders’ favourite market clusters during the past 20 trades \\ 
 \midrule
\multirow{4}{*}{Future} & NextTotalPL\_GBP20 & Continuous & P\&L for the next 20 trades in the future \\
 & NextTotalPL\_GBP & Continuous & P\&L for the next 100 trades in the future \\
 & WinningRate & Continuous & Winning rate of the next 100 trades in the future \\
 & TotalTrades & Continuous & Total trades made by a trader \\ 
 \midrule
\multirow{2}{*}{Identifier} & accountid & Discrete & The account ID of the trader \\
 & Period & Discrete & Buckets of 20 trades per account \\ 
 \bottomrule
\end{tabular}%
}
\caption{Comprehensive overview of dataset features, including type and description.}
\label{tab:feats_description}
\end{table}

\subsection{Data Preprocessing and Features Selection}
\label{sec:data-preprocess}

\subsubsection*{Train-Validation-Test Split}
The dataset is divided into training (70\%), validation (10\%), and testing (20\%) subsets. We ensure a minority-to-majority class ratio of 1\% to 99\% in all the training, validation and test splits to simulate realistic class distribution scenarios.

\subsubsection*{Ranking Group Allocation}

Ranking models require data to be grouped, typically using a \textit{qid} identifier, to ensure items are ranked only within their respective groups. Since our trading dataset was not originally structured for ranking tasks, we developed a custom grouping method to adapt it.
For training, we grouped traders based on three criteria: (1) market, (2) ensuring each group has at least one risky trader to handle label imbalance, and (3) the period in which trades occurred. 
For testing, we applied only criteria (1) and (3) to ensure a realistic and unbiased evaluation scenario. Algorithm~\ref{alg:rank_group_allo} details the group allocation process.
We evaluated multiple group sizes (20, 30, 50, 100, 200) to thoroughly test performance across varying contexts.

\begin{algorithm}
\caption{Ranking Group Allocation}
\label{alg:rank_group_allo}
\textbf{Input:} group\_size, data\_split

$\text{grouped\_data} \gets []$

\ForAll{market in market\_clusters}{

    $ \text{market\_data} \gets \text{data\_split}[\text{data\_split}[``\text{market}"]==1]$
    
    \ForAll{p in Period}{
        $\text{period\_market\_data} \gets \text{market\_data}[\text{market\_data}[``\text{Period}"]==p]$

        \eIf{data\_split is train set}
        {

            $\text{risky\_list} \gets \text{period\_market\_data}[\text{period\_market\_data}[``\text{label}"]==1]$

            $\text{normal\_list} \gets \text{period\_market\_data}[\text{period\_market\_data}[``\text{label}"]==0]$
    
            \While{(risky\_list is not empty) and (normal\_list is not empty)}{
                $\text{risky} \gets \text{risky\_list.sample(min(len(risky\_list), 1))}$\;
                
                risky\_list.remove(risky)
                
                $\text{normal} \gets \text{normal\_list.sample}(\min(\text{len(normal\_list)}, \text{group\_size-1)})$
    
                normal\_list.remove(normal)
    
                grouped\_data.append(concat(risky, normal))
        
                }

        }{
            $\text{sample\_num} \gets \text{min(len(period\_market\_data), group\_size-1)}$
            
            $\text{sample\_list} \gets \text{period\_market\_data.sample(sample\_num)}$

            grouped\_data.append(sample\_list)
        
        }
        
    }
    
}

\textbf{Return} grouped\_data

\end{algorithm}

\subsubsection*{Feature Selection}

\begin{table}[h]
\centering
\begin{tabular}{lcccccc}
\toprule
\multirow{2}{*}{Features} & \multicolumn{2}{c}{Mean} & \multicolumn{2}{c}{Std. Dev.} &  \multicolumn{2}{c}{Skew}  \\
 & Normal & Risky & Normal & Risky & Normal & Risky \\ \midrule
AVGPTS3\_20 & 0.430 & 0.588 & 0.211 & 0.244 & 0.331 & -0.211 \\
AvgOpen20 & 0.535 & 0.638 & 0.220 & 0.345 & -0.337 & -0.627 \\
AvgShortSales20 & 0.485 & 0.419 & 0.270 & 0.330 & -0.025 & 0.258 \\
DurationRate20 & 0.320 & 0.355 & 0.120 & 0.132 & -0.151 & -0.464 \\
DurationRatio20 & 0.127 & 0.166 & 0.067 & 0.124 & 3.492 & 3.821 \\
PassAvgReturn & 0.502 & 0.540 & 0.053 & 0.121 & -0.305 & 0.736 \\
ProfitRate20 & 0.497 & 0.623 & 0.243 & 0.297 & 0.343 & -0.325 \\
ProfitxDur20 & 0.327 & 0.422 & 0.173 & 0.223 & 0.988 & 0.270 \\
SharpeRatio20 & 0.443 & 0.489 & 0.082 & 0.127 & 1.096 & 0.743 \\
WinTradeRate20 & 0.623 & 0.685 & 0.204 & 0.238 & -0.204 & -0.557 \\
PerFTSE20 & 0.249 & 0.157 & 0.356 & 0.279 & 1.163 & 1.906 \\
TradFQ20 & 0.363 & 0.314 & 0.292 & 0.285 & 1.153 & 1.394 \\
OrderCloseRate20 & 0.182 & 0.189 & 0.263 & 0.286 & 1.464 & 1.476 \\
NumTrades & 0.305 & 0.270 & 0.326 & 0.290 & 1.094 & 1.290 \\
\bottomrule
\end{tabular}%
\caption{Summary of Selected Features: key statistics and descriptions for each chosen feature, where features denoted with the suffix ``20'' indicate metrics derived from the traders' most recent 20 transactions.}
\label{tab:feats_stats}
\end{table}

Acknowledging the importance of feature engineering in deep learning models \citep{airbnb}, we meticulously selected features for our LETOR algorithms. 
For models that do not originally support categorical features, we specifically excluded attributes such as age group, market preferences, and segment group. 
Our preliminary evaluations indicated that even converting these categorical features to dummy variables resulted in a performance decline for models such as Rankformer~\cite{rankformer}. 
However, for other models in our study that can handle categorical data, these features were retained. 
For the 14 continuous features in Table~\ref{tab:feats_description}, we provide the distributional differences analysis between normal and risky traders, detailed in Table 2. 

\section{Problem definition}

Specifically, our decision task is to determine whether the market maker should hedge an upcoming trade $j$ from trader $i$ and monitor the traders' future activities, as shown in Equation~\eqref{eq:hedge_decision}.

\begin{equation}
\label{eq:hedge_decision}
    y_{ij} = 
    \begin{cases}
        1, & \text{hedge} \\
        0, & \text{not hedge}
    \end{cases}
\end{equation}

In the training phase of our supervised ranking models, accurately labelling the dataset poses a significant challenge due to two primary reasons. 
Firstly, there are considerable legal complexities associated with financial institutions sharing sensitive information regarding past trading anomalies and breaches with researchers. 
Secondly, the reliability of true labels is compromised by noise, which stems from the incomplete identification of risky traders.
This results in a scenario where a subset of genuinely risky traders are inaccurately labelled as normal, thereby diluting the integrity of the labels \cite{10.1145/3459637.3482433}.

To overcome these labelling challenges, we categorise traders exhibiting extremely high profitability as risky from the perspective of market makers, who are ultimately liable for the losses incurred due to these traders' activities. 
This labelling method is based on the understanding that traders with significantly high returns present a considerable risk to market makers by potentially exploiting market weaknesses for their gain. 
Importantly, this strategy also helps avoid the problem of noisy labels since returns directly reflect trading outcomes, providing a clear indicator of trader performance and risk level, which is aligned with \cite{dnn}.
Consequently, identifying these highly profitable traders aids market makers in refining their hedging strategies, enabling more effective risk management practices in the face of uncertain market dynamics. 

Formally, we define the risky label $y_{ij}$ through Equations \eqref{eq:return} to \eqref{eq:labels}. 
For a trader indexed by $i$ at the issuance of trade indexed by $j$, profitability is assessed over the subsequent 100 trades. 
Specifically, the future return for trader $i$ at trade $j$ is calculated as:

\begin{align}
\label{eq:return}
\text{Return}_{i,j} & = \frac{\sum_{j < k \leq j+100} P\&L_{i,k}}{\sum_{j < k \leq j+100} \text{Margin}_{i,k}},
\end{align}

where $P\&L_{i,k}$ denotes the profit and loss realised by trader $i$ on trade $k$, and $\text{Margin}_{i,k}$ represents the margin requirement imposed by the market maker for trader $i$ on trade $k$, computed as the product of the stake size and the required margin percentage.

Traders are labelled risky if their returns rank within the top $\alpha\%$ among all traders in the subsequent hundred trades following trade $j$, defined formally as:

\begin{align}
\label{eq:top_alpha}
\text{Top}_\alpha (x, \mathcal{X}) & =
\begin{cases}
1, & \text{if } x \text{ ranks within the top } \alpha\% \text{ of } \mathcal{X} \\
0, & \text{otherwise}
\end{cases} \\
\label{eq:labels}
y_{ij} & = \text{Top}_\alpha(\text{Return}_{i,j}, {\text{Return}_{i',j'} \text{ for all traders } i' \text{ at all trade} j'})
\end{align}

In summary, the label $y_{ij}$ indicates whether trader $i$ is categorised as risky at the time trade $j$ is issued, based on their subsequent profitability relative to peers.

We focus on an extreme scenario where $\alpha = 1$\%. 
This choice allows us to examine traders who potentially engage in abnormal activities and aim to generate significant gains within a limited number of trades \citep{CHESNEY2015263}.

\paragraph{Being Lucky?} 

We believe our label definition mitigates the impact of traders achieving high profits by chance from four perspectives. 
Firstly, using a past 20-trade window and examining the subsequent 100 trades to identify risk aligns with the strategy employed by the data provider's dealing desk \cite{dnn}, which has proven effective in real-world applications. 
Secondly, focusing on the top 1\% of returns aligns with the ``extreme value theory'' in finance, which studies extreme deviations from the median of probability distributions. 
This approach is particularly useful for identifying outliers who may pose significant risks to market stability \cite{10.1093/rfs/hhg058extremevaluedependence}. 
By concentrating on these extreme cases, we aim to identify traders whose performance significantly deviates from the norm, thereby flagging potential anomalies that warrant closer scrutiny.
Thirdly, our processed data incorporates key control variables such as the number of trades (\textit{TradFQ20}), average return on past trades (\textit{PassAvgReturn}), and win rate (\textit{WinTradeRate20}), which serve as proxies for traders' experience and skill levels \cite{kumarfeatureimportance}. 
By accounting for these factors, we control for the inherent variability in trader proficiency, thereby ensuring that the identification of risky traders is not confounded by differences in experience or strategic competence. 
This comprehensive feature set allows the model to discern whether high returns are attributable to skillful trading or merely stochastic fluctuations.
Lastly, even if some traders achieve extremely high profits purely by luck, it is crucial to notice this so that market makers can plan their hedging strategies in advance to reduce losses.

\section{Methodology}
\label{sec:methodology}

Integrating a ranking algorithm to predict risky traders presents several challenges (CHs):

\begin{itemize}
    \item \textbf{CH1:} There is a significant class imbalance in risky trader prediction, complicating the task.
    \item \textbf{CH2:} Existing ranking methods \cite{lambdarank,rankformer} lack adaptation for financial metrics, making it difficult to prioritise extremely risky traders.
    \item \textbf{CH3:} The robustness of deep learning-based ranking models is limited, as highlighted in \cite{NEURIPS2022_0378c769,google-neural-outperform-gbdt}.
\end{itemize}

To address \textbf{CH1} and \textbf{CH2}, we propose the Profit-Aware Binary Cross Entropy (PA-BCE) loss. 
For \textbf{CH1}, we first demonstrate that a LambdaRank-like objective \citep{lambdarank} based on binary cross-entropy (BCE) loss effectively mitigates class imbalance. 
Building upon this insight, PA-BCE modifies the LambdaRank approach, which originally weights pairwise comparisons according to generic metric changes, by explicitly incorporating pairwise profit differences (P\&L gaps). Unlike standard LambdaRank, our method employs a customised financial metric defined by the logarithmic transformation of pairwise P\&L gaps, highlighting significant differences in trader risk. 
This targeted adjustment not only further alleviates class imbalance through balanced pairwise comparisons but also aligns the ranking objectives directly with key financial criteria, effectively addressing \textbf{CH2}.

For \textbf{CH3}, the lack of robustness arises from an incomplete representation of both individual feature dependencies within each trader’s profile and the interactions across traders \cite{huang2020tabtransformer,NEURIPS2021_9d86d83f}. 
Existing ranking models like Rankformer \cite{rankformer}, which rely solely on cross-item (trader) attention, capture only inter-trader comparisons and overlook intra-trader dependencies, which are critical interactions among features within a single trader’s profile.
To address this, we introduce a Self-Cross-Trader attention pipeline that integrates FT-Transformer embeddings \cite{NEURIPS2021_9d86d83f}, known for enhancing deep learning performance on tabular data. Self-trader attention uses contextual embeddings \cite{devlin-etal-2019-bert} to capture intra-trader or feature-level relationships. In contrast, cross-trader attention models pairwise comparisons between traders' contextual embeddings, capturing cross-trader relationships. This pipeline improves the model's ability to analyse individual behaviours and rank multiple traders in the financial market, thereby enhancing the robustness of deep learning-based models for predicting risky traders.

Figure~\ref{fig:model_arch} presents the complete pipeline, integrating the PA-BCE loss (Section~\ref{sec:PA-BCE}) and Self-Cross-Trader Attention mechanisms (Section~\ref{sec:self-cross-attention}). 

\begin{figure}[htbp]
\centering
\includegraphics[width=\textwidth]{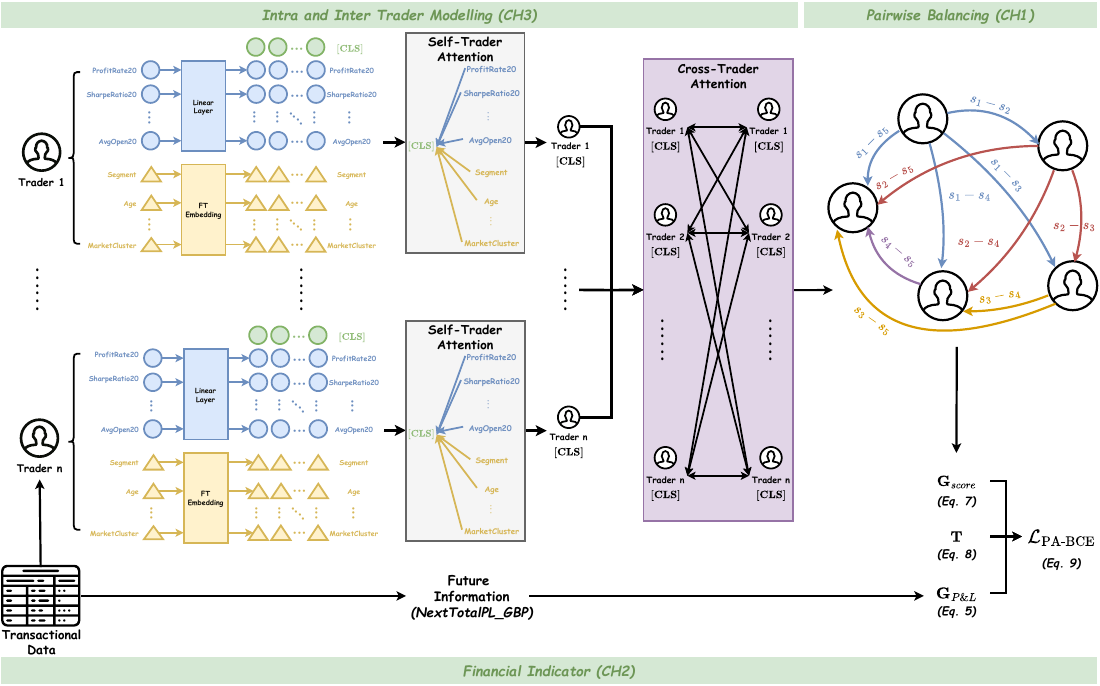}
\caption{Architecture of the PA-RiskRanker pipeline. Continuous features are processed through a linear layer, while categorical features are encoded using FT-Embedding. The \texttt{[CLS]} token in the self-trader attention mechanism captures feature dependencies, generating contextual embeddings. Cross-trader attention then analyses trader interrelationships to detect risky patterns, forming the $\mathbf{G}_{score}$ matrix. The pipeline utilises the \textit{NextTotalPL\_GBP} information for $\mathbf{G}_{P\&L}$ matrix construction and is trained with PA-BCE loss for precise ranking alignment.}
\label{fig:model_arch}
\end{figure}

\subsection{Profit-Aware BCE}
\label{sec:PA-BCE}

The primary challenge in our task is that existing ranking algorithms are designed for recommendation \cite{10.1145/3477495.3531931,peiPersonalizedRerankingRecommendation2019,rank-insights} and information retrieval systems \cite{ai2018learning,lambdarank,rankformer}. 
These tasks differ from ranking risky traders in two aspects. 

Firstly, severe class imbalance, with typically only about 1\% of traders being risky, poses substantial difficulties for listwise ranking methods, particularly those based on deep learning, as highlighted in prior studies \cite{surveyonclassimbalance} and confirmed by our preliminary results in Appendix~\ref{sec:preliminary} (\textbf{CH1}). 
Secondly, while pairwise ranking methods naturally mitigate imbalance by comparing items in pairs (see Proposition~\ref{proposition:balance}), standard approaches like LambdaRank \citep{lambdarank} do not account for the magnitude of financial risk differences, which can be explicitly quantified by traders' future profits (\textbf{CH2}).

To address these limitations, we propose the \textbf{Profit-Aware Binary Cross Entropy (PA-BCE)} loss. 
While leveraging the advantages of the LambdaRank-like pairwise objective, which is less sensitive to class imbalance, PA-BCE integrates multiple elements explicitly tailored to the financial domain. 
Specifically, we introduce a profit-based pairwise gap metric $\mathbf{G}_{P\&L}$ to replace the metric-change term ($\Delta$) in LambdaRank, represented through a graph-based relational structure \citep{gnn-paper} among traders. 
Each trader corresponds to a node in a fully connected graph, and edges explicitly encode financial risk differences.
Furthermore, given the vast numerical range of traders' profits, spanning from severe losses to substantial gains, we introduce a logarithmic transformation to the pairwise P\&L gaps. 
This logarithmic transformation ensures meaningful distinctions across extreme profit variations. 
To further stabilise training and enhance performance, we implement a top-K sampling strategy that explicitly prioritises the most financially significant traders.

\paragraph{Graph Construction and Loss Computation.}

Consider a ranking group $g$ consisting of $n$ traders, ordered by their true profits over the subsequent 100 trades (\textit{NextTotalPL\_GBP}), as recommended by prior studies \cite{dnn}. Each trader has embedded features from their 20 most recent trades, represented as $\{\mathbf{x}_1, \dots, \mathbf{x}_n \mid \mathbf{x}_i \in \mathbb{R}^d\}$.

In Equation~\eqref{eq:pnl_gap}, we define the P\&L gap matrix $\mathbf{G}_{P\&L}^{(g)}$, which functions as the customised metric-change term in the LambdaRank-like loss. Here, $p_i$ and $p_j$ represent the actual profits of traders $i$ and $j$, respectively. The logarithmic transformation is essential due to the extensive range of P\&L values in the dataset, which spans from $-1$ million to over 460K. This transformation normalises extreme profit variations, ensuring meaningful differentiation across all profit ranges while maintaining numerical stability. 

\begin{equation} \label{eq:pnl_gap} 
\mathbf{G}_{P\&L}^{(g)} = 
\begin{pmatrix} 
0 & \log(1 + p_1 - p_2) & \cdots & \log(1 + p_1 - p_n) \\ 
\log(1 + p_1 - p_2) & 0 & \cdots & \log(1 + p_2 - p_n) \\ 
\vdots & \vdots & \ddots & \vdots \\ 
\log(1 + p_1 - p_n) & \log(1 + p_2 - p_n) & \cdots & 0 
\end{pmatrix} 
\end{equation}

Given the predicted scores $(s_1, \dots, s_n)$ from the ranking model, we construct the standard LambdaRank score gap matrix, denoted as $\mathbf{G}_{\text{score}}^{(g)}$, as shown in Equation~\eqref{eq:g_score}. 
Here, $\sigma(\cdot)$ represents the sigmoid function.

\begin{equation} \label{eq:g_score} 
\mathbf{G}_{\text{score}}^{(g)} = 
\begin{pmatrix} 
0 & \sigma(s_1 - s_2) & \cdots & \sigma(s_1 - s_n) \\ 
\sigma(s_2 - s_1) & 0 & \cdots & \sigma(s_2 - s_n) \\ 
\vdots & \vdots & \ddots & \vdots \\ 
\sigma(s_n - s_1) & \sigma(s_n - s_2) & \cdots & 0 
\end{pmatrix} 
\end{equation}

The target matrix $\mathbf{T}^{(g)}$ indicating ground-truth rankings is defined as:

\begin{equation}
\label{eq:target_matrix}
\mathbf{T}^{(g)} =
\begin{pmatrix}
0 & 1 & \cdots & 1 \\
0 & 0 & \cdots & 1 \\
\vdots & \vdots & \ddots & \vdots \\
0 & 0 & \cdots & 0
\end{pmatrix}.
\end{equation}

For a set of $S$ ranking groups, the PA-BCE loss is calculated as:

\begin{equation}
\label{eq:loss}
\mathcal{L}_{\text{PA-BCE}} = \sum_{g=1}^{S} \sum_{i,j} \mathbf{G}_{P\&L}^{(g)}(i,j) \cdot \text{BCE}\left(\mathbf{G}_{\text{score}}^{(g)}(i,j), \mathbf{T}^{(g)}(i,j)\right),
\end{equation}
where BCE is the binary cross-entropy loss. 

In practice, due to the symmetry of the P\&L gap matrix and the properties of the loss function, we compute the loss over the upper triangular part of the matrices, as optimising over the upper triangular part is equivalent to optimising over the entire matrices (see Proposition~\ref{proposition:upper_triangular_equivalence}). 
This is equivalent to having a directed graph constructed from the ranking group such that traders with higher risks point to those with lower risks as demonstrated in Figure~\ref{fig:model_arch}.

\paragraph{Top-K Sampling Method.} 
Managing rankings of over 100 traders is challenging, as larger group sizes typically lead to training instability. Financial markets are inherently noisy, and increasing the number of traders introduces additional complexity, hindering model convergence \cite{RePEc:kap:fmktpm:v:34:y:2020:i:4:d:10.1007_s11408-020-00368-y,bouchaud2009theory}. To mitigate this, we apply top-K sampling, concentrating on the most profitable traders. This method, widely adopted in large-scale ranking problems \cite{pmlr-v31-eriksson13a,10.5555/3122009.3176847}, reduces noise by emphasising meaningful distinctions among top performers, resulting in more stable and effective training.

In summary, the proposed PA-BCE loss utilises graph modelling to incorporate profit information, enhancing precise risk quantification in ranking traders. 
This approach overcomes the limitations while maintaining the strengths of pairwise ranking on imbalanced financial data.

\subsection{Self-Cross-Trader Attention Pipeline}
\label{sec:self-cross-attention}

Recall that the last significant challenge (\textbf{CH3}) is the limited robustness of deep learning-based ranking models on the domain-specific financial risk prediction task. 
To tackle the issue of lacking intra-trader modelling, we propose the Self-Cross-Trader Attention pipeline that integrates both self-trader attention for intra-trader feature analysis and cross-trader attention for inter-trader comparisons. 
This dual approach enhances robustness by enabling the model to capture complex intra-trader relationships alongside competitive inter-trader dynamics.

\paragraph{Feature Embedding: Processing Continuous and Categorical Features.}

In the initial feature embedding stage, continuous and categorical features are processed separately to capture the unique characteristics of each type. 
For each trader, we transform the input features vector $\mathbf{x}$ into embeddings $\mathbf{E}_{\text{trader}} \in \mathbb{R}^{k \times d_k}$, where $k$ is the total number of features (continuous and categorical) and $d_k$ is the embedding dimension. The embedding for a given feature $x_j$ is computed as follows:

\begin{equation} 
    \mathbf{e}_j = \mathbf{b}_j + \mathbf{W}_j x_j \in \mathbb{R}^{d_k},
\end{equation}

where $\mathbf{W}_j \in \mathbb{R}^{d}$ is a weight vector that linearly projects the feature $x_j$ into a $d$-dimensional embedding space, and $\mathbf{b}_j$ is the corresponding bias vector. 
The embedding function $f_j$ thus applies a simple linear transformation to project the scalar feature $x_j$ into a $d$-dimensional representation. 

The embeddings are computed differently for continuous and categorical features:

\begin{itemize}
    \item \textbf{Continuous Features:} For a continuous feature $x_j^{\text{(num)}}$, the embedding uses a specific weight vector $\mathbf{W}_j^{\text{(num)}} \in \mathbb{R}^{d_k}$ and bias term $\mathbf{b}_j^{\text{(num)}}$:
    \begin{equation} 
        \mathbf{e}_j^{\text{(num)}} = \mathbf{b}_j^{\text{(num)}} + \mathbf{W}_j^{\text{(num)}} x_j^{\text{(num)}}.
    \end{equation}
    
    \item \textbf{Categorical (Discrete) Features:} Categorical features are handled differently due to their discrete nature. 
    For each categorical feature $x_j^{\text{(cat)}}$, there is a finite set of possible values.
    For example, a categorical feature like ``MarketCluster'' can have values 0 to 9 as shown in Table~\ref{tab:feats_description}.
    The number of distinct categories for a feature is called its vocabulary size, denoted $S_j$.
    Similar to \cite{NEURIPS2021_9d86d83f}, we use an embedding lookup table $\mathbf{W}_j^{\text{(cat)}} \in \mathbb{R}^{S_j \times d_k}$ to embed these categorical values into the $d_k$-dimensional space.
    Each row in this lookup table corresponds to a unique category in the feature's vocabulary and provides a learnable $d_k$-dimensional embedding for that category. 
    The embedding for a categorical feature $x_j^{\text{(cat)}}$ is computed as:
    \begin{equation}
        \mathbf{e}_j^{\text{(cat)}} = \mathbf{b}_j^{\text{(cat)}} + \mathbf{e}_j^\intercal \mathbf{W}_j^{\text{(cat)}} \in \mathbb{R}^{d_k},
    \end{equation}
    where $\mathbf{e}_j$ is a one-hot vector representing the categorical feature $x_j^{\text{(cat)}}$, and $\mathbf{b}_j^{\text{(cat)}}$ is the bias term for the categorical feature.
\end{itemize}

The final feature embeddings $\mathbf{E}_{\text{trader}}$ for each trader are obtained by stacking the continuous and categorical embeddings:

\begin{equation}
    \mathbf{E}_{\text{trader}} = \text{stack}\left[ \mathbf{e}_1, \mathbf{e}_2, \dots, \mathbf{e}_k \right] \in \mathbb{R}^{k \times d_k},
\end{equation}

where $\mathbf{e}_j$ represents either $\mathbf{e}_j^{\text{(num)}}$ or $\mathbf{e}_j^{\text{(cat)}}$, depending on the type of feature. This stacked representation $\mathbf{E}_{\text{trader}}$ forms the unified feature embedding for each trader, which is then used as input to the self-trader attention mechanism in the following stage.

\paragraph{Self-Trader Attention: Capturing Intra-Trader Dependencies.}

After embedding, a \texttt{[CLS]} token is appended to the sequence to aggregate information across all features. The extended input sequence for self-trader attention becomes:

\begin{equation}
    \mathbf{E}_{\text{trader}}^{\text{[CLS]}} = \text{stack}[\text{[CLS]}, \mathbf{e}_1, \mathbf{e}_2, \dots, \mathbf{e}_k],
\end{equation}

where \texttt{[CLS]} is initialised as a learnable vector. The \texttt{[CLS]} token, along with the embedded features, is then processed through self-attention. 
Self-trader queries $ Q_s $, keys $ K_s $, and values $ V_s $ are computed as follows:

\begin{equation}
    Q_s = \mathbf{E}_{\text{trader}}^{\text{[CLS]}} W_s^Q, \quad K_s = \mathbf{E}_{\text{trader}}^{\text{[CLS]}} W_s^K, \quad V_s = \mathbf{E}_{\text{trader}}^{\text{[CLS]}} W_s^V,
\end{equation}

where $ W_s^Q $, $ W_s^K $, and $ W_s^V $ are the weight matrices for the self-attention mechanism. 
Self-trader attention is derived from the original self-attention \cite{attention} mechanism as follows:

\begin{equation}
\label{eq:self-attention}
    \text{Attention}(Q_s, K_s, V_s) = \text{softmax}\left(\frac{Q_s K_s^\intercal}{\sqrt{d_k}}\right)V_s,
\end{equation}

The \texttt{[CLS]} token accumulates information from the entire feature sequence, creating a condensed representation of the trader’s profile that captures intra-trader dependencies. 
This is a well-established technique in modern language models like BERT \cite{devlin-etal-2019-bert}, and it has been proven to be effective in gathering sequence-level information. 
The enhanced [CLS] representation is then passed to the cross-trader attention mechanism.

\paragraph{Cross-Trader Attention: Modelling Inter-Trader Relationships.}

Once the \texttt{[CLS]} token for each trader has been refined through self-attention, these tokens are used in the cross-trader attention stage. Here, the \texttt{[CLS]} tokens from different traders serve as inputs, allowing the model to capture competitive dynamics by comparing each trader’s summary representation with others.

Given two \texttt{[CLS]} representations from different traders, denoted as $[\text{CLS}]_i$ and $[\text{CLS}]_j$, cross-trader attention is computed as:

\begin{equation}
    \text{Cross-Trader Attention}([\text{CLS}]_i, [\text{CLS}]_j) = \text{softmax}\left(\frac{[\text{CLS}]_i W_c^Q ([\text{CLS}]_j W_c^K)^\intercal}{\sqrt{d_k}}\right) [\text{CLS}]_j W_c^V.
\end{equation}

where $ W_c^Q $, $ W_c^K $, and $ W_c^V $ are the weight matrices for the cross-trader attention. 
This operation enables each trader’s \texttt{[CLS]} token to interact with other traders’ \texttt{[CLS]} tokens, capturing inter-trader relationships such as shared risk patterns or competitive influences.

Our Self-Cross-Trader Attention pipeline allows for a comprehensive analysis of both intra-trader and inter-trader dimensions, capturing the individual behaviour of traders as well as their interactions within a group. 

\paragraph{Training Strategy.}

Inspired by advancements in computer vision and natural language processing (NLP), we further employ a two-step training strategy to enhance the representation of traders and improve the robustness of our Self-Cross-Trader Attention pipeline. 
In the first step, we pretrain the contextual \texttt{[CLS]} embedding, derived from self-trader attention, on the risky trader classification task using the FT-Transformer architecture \citep{NEURIPS2021_9d86d83f}. 
This pretraining stage enables the model to gains an initial understanding of essential intra-trader dependencies and feature interactions unique to each trader’s profile, thereby reducing the complexity of the downstream ranking task \cite{wang-etal-2020-pretrain}.

Once pretraining is complete, we incoporate these enriched \texttt{[CLS]} embeddings with the cross-trader attention mechanism to model inter-trader relationships. 
During this stage, we fine-tune the embeddings using the PA-BCE loss to align them with the specific objectives of the ranking task.
This fine-tuning step ensures that the embeddings are further refined to capture competitive dynamics and relative risk patterns between traders, which are essential for accurate ranking.

\section{Experiments}
\label{sec:exp_setup}

This section details three comprehensive sets of evaluations designed to test the benchmark models.

Our first evaluation focuses on ranking performance. Here, we specifically assess the effectiveness of our PA-BCE loss and the enhancements brought by the refined architecture. This phase aims to elucidate the impact of our modifications on the model's ability to accurately rank traders in terms of risk.

Secondly, given that ranking models are not customarily assessed using classification metrics, we operate the evaluation under the prior knowledge that 1\% of the traders in our dataset are risky.
This setup allows us to evaluate the ranking models using classification and anomaly detection metrics. 
To validate that this setup does not inherently favour ranking models, we also conduct traditional prediction experiments without prior knowledge of the proportion of risky traders.

Finally, we implement a two-step evaluation process, aligning with the practical need for interpretability. 
This allows us to assess how effectively the benchmark models with less interpretability contribute to improving the overall performance of the interpretable models, crucial for real-world financial decision-making.

Our primary experiments introduced in this section utilise the private financial transaction dataset as mentioned in Section~\ref{sec:dataset}. To ensure reproducibility, we also evaluate our approach using two publicly available datasets\footnote{The codes and preprocessed data are available at \url{https://github.com/waylonli/PARiskRanker}.}: one for credit card fraud detection and another for job profit prediction. We adapt these original tasks to align with our scenario. 
The results demonstrate the transferability of our proposed approach. 
Additional details can be found in Appendix~\ref{appendix:reproducible-exp}.

\subsection{Metrics} 
\label{sec:metrics}

We employ two sets of metrics in this paper respectively for evaluating ranking and classification performance. 

To measure the performance of the ranking algorithms, we use Normalised Discounted Cumulative Gain (NDCG) and Mean Reciprocal Rank (MRR) which are two of the most popular and representative metrics for evaluating any ranking model \cite{lambdarank,pang2020setrank,rankformer}.

\paragraph{Normalised Discounted Cumulative Gain.}
NDCG (Equation~\eqref{eq:ndcg}) considers the entire ranked list and penalises the scenario where risky traders are assigned lower ranks.
When the dataset includes graded relevance values, NDCG is an appropriate choice as it incorporates relevance information.

\begin{align}
\textrm{DCG}p & = \sum{i=1}^p \frac{2^{rel_i}-1}{\log_2 (i+1)} \\
\textrm{IDCG}p & = \sum{i=1}^{|REL_p|} \frac{2^{rel_i}-1}{\log_2 (i+1)} \\
\textrm{nDCG}_p & = \frac{\textrm{DCG}_p}{\textrm{IDCG}_p}
\label{eq:ndcg}
\end{align}

\paragraph{Mean Reciprocal Rank.}

MRR (Equation~\eqref{eq:mrr}) is another critical metric for assessing ranking models, in our case, emphasising the model's effectiveness in identifying the first occurrence of a risky trader in the ranking group. It is particularly insightful for evaluating how promptly a model can highlight the highest-risk traders at the top of the ranking.

\begin{equation}
\label{eq:mrr}
    \textrm{MRR} = \frac{1}{|Q|} \sum_{j=1}^{|Q|} \frac{1}{\text{rank}_j}
\end{equation}

\paragraph{Classification Metrics.}
For classification models, standard evaluation metrics such as specificity, precision, macro $\text{F}_1$ score, and Area Under the Curve (AUC) are typically employed. 
The omission of accuracy as a metric is deliberate. Even a naive baseline that predicts everyone as normal traders can achieve a high accuracy of 99\% on the whole dataset, given that only the top 1\% of data with the highest profit rates are labelled as anomalies.

The aforementioned metrics, while effective in many scenarios, may not fully capture the nuances of financial data analysis. Particularly, they often fall short in addressing the asymmetrical risk-reward profile inherent in financial contexts.
To bridge this gap and align our evaluation more closely with our specific goals, we incorporate the P\&L metric. This metric provides a more direct measurement of financial performance, reflecting the actual economic impact of the model's predictions and ensuring a more nuanced and relevant assessment for our financial tasks \cite{dnn}.
In the financial market, where trade frequency can exceed 25 million observations per minute \cite{https://doi.org/10.1111/j.0306-686X.2004.00552.x}, even small gains, such as £0.1 per 20 trades, can accumulate into substantial profits due to the high transaction volume.

\paragraph{P\&L}

We aim to simulate the decision-making process of market makers faced with the choice to hedge or not hedge an imminent trade $j$ from trader $i$.
The decision to hedge is informed by the predicted label $y_{ij} \in \{0,1\}$, as articulated in Equation~\eqref{eq:hedge_decision}. 
Hedging a trade, when done completely, is assumed to neutralize the market maker's financial exposure, resulting in neither profit nor loss from that specific trade. 
In contrast, failing to hedge a trade exposes the market maker to the financial outcome of the trade: absorbing losses equivalent to the trader’s profits or gaining profits equivalent to the trader's losses.

Once trader $i$ has been marked as risky, the market makers will continue hedging the following 20 trades (including the current trade) and re-evaluate the status of this trader after that. Therefore, the metric is calculated using the P\&L of the future 20 trades, denoted as $\text{NextProfit\_20}$. 
For trades labelled as non-high-risk ($y_i = 0$), the market maker's P\&L reflects the inverse of trader $i$'s subsequent 20 trades' outcomes. 
In contrast, for trades identified as high-risk ($y_i = 1$), hedging actions render the P\&L neutral, indicating no financial gain or loss, as detailed in Equation~\eqref{eq:pnl}. 
This method abstracts away from transactional costs for simplicity.

\begin{equation}
\label{eq:pnl}
    P\&L = \sum_{i,j} -(1-y_{ij}) * \textrm{NextProfit\_20}_{ij}
\end{equation}

\subsection{Benchmark Models}

Our task encompasses multiple perspectives: ranking, classification, and anomaly (outlier) detection. Consequently, we included three distinct groups of benchmark models respectively.

\paragraph{Ranking Benchmarks.}

We evaluate PA-RiskRanker against a comprehensive suite of ranking baselines. 
First, we include Rankformer \cite{rankformer}, a representative transformer-based ranking model.
We also consider LambdaLoss \cite{lambdaloss} as an alternative objective for Rankformer, given its effectiveness in optimising ranking metrics directly.
We also incorporate Surrender on Outliers and Rank (SOUR) \cite{sour}, which explicitly handles outlier items within queries to boost overall ranking quality.
In addition to neural methods, we benchmark against traditional gradient-boosted ranking via LambdaMART, widely regarded as a strong classical baseline \cite{lambdarank,pang2020setrank,rankformer}.
We employ the implementation of LambdaMART via \texttt{LGBMRanker} from the LightGBM library \cite{NIPS2017_6449f44a} with the standard \textit{LambdaRank} objective.
Finally, to assess the impact of our profit-aware objective, we replace the canonical $\Delta$ in \textit{LambdaRank} with our proposed $\mathbf{G}_{P\&L}$ from Section~\ref{sec:PA-BCE}, resulting in PA-$\lambda$MART, using the customised objective functionality in the LightGBM library.

We claim that these serve as adequate ranking benchmarks.
While GNNRank \cite{GNNRank} was considered, it was deemed unsuitable due to its focus on recovering global rankings from historical data rather than predictive modelling. 
GNNRank performs well at determining global rankings, such as football team standings based on past game outcomes, but it cannot predict future rankings or performance changes. Consequently, it was excluded. 
Additionally, Rankformer was chosen over SetRank \cite{pang2020setrank,rankformer} because it is an upgraded version of SetRank, making the inclusion of both redundant.
We also reviewed several recent neural ranking approaches from the information retrieval literature, each making valuable contributions to their domains. 
The self-distilled learning to rank method~\cite{10.1145/3681784} advances document retrieval by improving robustness through down-weighting uncertain or noisy samples, whereas our goal is to explicitly identify and prioritise outliers (risky traders). 
Counterfactual learning to rank via knowledge distillation~\cite{counterfactualletor} addresses bias and generalisation in settings with implicit feedback like click data, which is fundamentally different from our supervised, profit-based context. 
Distilled neural networks for efficient ranking~\cite{10.1109/TKDE.2022.3152585} focus on compressing tree ensembles for computational efficiency, but efficiency is not the primary concern in our risk-focused financial applications. Thus, while impactful in their original domains, these methods are not directly applicable as benchmarks for our task.

\paragraph{Classification Benchmarks.}

Given that the task of detecting risky traders can be approached as a binary classification problem, 
we incorporate the complete FT-Transformer model \cite{NEURIPS2021_9d86d83f}, regarded as a ``best case'' scenario for deep learning applications on tabular data \cite{NEURIPS2022_0378c769}. 
We also incorporate the \verb|XGBClassifier|\footnote{\url{https://xgboost.readthedocs.io/en/stable/python/python_api.html}} from XGBoost \cite{xgboost}, especially relevant due to the performance enhancements in its 2.0 version released in September 2023. 
Additionally, the \verb|LGBMClassifier|\footnote{\url{https://lightgbm.readthedocs.io/en/latest/pythonapi/lightgbm.LGBMClassifier.html}} from LightGBM \cite{NIPS2017_6449f44a} and the Random Forest classifier are also part of our model lineup, offering a comprehensive evaluation of classification approaches.

\paragraph{Anomaly Detection Benchmarks.}

Anomaly detection offers another perspective for modelling the risky trader detection task. This approach is prevalent in various domains, such as emerging disease detection, financial fraud identification, and fake news detection \cite{10.1145/3580305.3599302}. Anomaly detectors focus on learning the general distribution of normal data and identifying instances that significantly deviate from this majority pattern.

We select a range of recent anomaly detection models tailored for tabular data. These models encompass both unsupervised and weakly-supervised methods, including Deep SAD (weakly supervised) \cite{ruff2020deep}, FeaWAD (weakly-supervised) \cite{Zhou2021FeatureEW}, SLAD (unsupervised) \cite{10.5555/3618408.3620019}, and DIF (unsupervised) \cite{xu2023deep}. For practical implementation, we utilise the DeepOD\footnote{\url{https://deepod.readthedocs.io/en/latest/}} library \cite{xu2023deep}, which features 27 deep outlier detection / anomaly detection algorithms designed for a variety of applications.

\paragraph{Technical Setup.}

To ensure that all benchmark models were rigorously trained and assessed, we engaged in comprehensive hyperparameter optimisation to achieve adequate convergence in terms of loss minimisation.
The exhaustive list of hyperparameters utilised for each model is delineated in Appendix~\ref{appendix:hyper-parameters}. 
In configuring the FT-Transformer, we aligned the number of transformer layers to correspond with the number of blocks present in both the Rankformer and our PA-RiskRanker, ensuring architectural consistency.


\subsection{Ranking Performance}
\label{sec:rank_perform}

In alignment with the preliminary experiments detailed in Appendix~\ref{sec:preliminary}, we train our benchmark ranking models with various training group sizes, maintaining a consistent test group size of 100. 
The training group size is treated as a hyperparameter, and we select the optimal combination for each model to showcase its performance.

For the LambdaMART model, we employ configurations with both \(10^3\) and \(10^4\) trees, aligning with the setup in \cite{rankformer}. 
This is crucial, as advanced neural ranking methods like SetRank and Rankformer have faced challenges in outperforming LambdaMART even with \(10^3\) trees \cite{pang2020setrank,rankformer}.

Furthermore, the incorporation of the Self-Cross Trader Attention pipeline introduces an inconsistency in model size. 
To address this, we adapt the original Rankformer architecture with \textit{LogSoftmax} loss to include 6 layers of transformer encoders. 
This configuration is modified to 2 layers dedicated to the self-trader attention and 4 layers for cross-trader attention in PA-RiskRanker, ensuring a fair comparison that accounts for the added complexity of the intra-trader modelling.

\begin{table}[h!]
\centering
\resizebox{0.9\columnwidth}{!}{%
\begin{tabular}{lrrrr}
\toprule
Ranking Model & NDCG@3 & NDCG@5 & NDCG@10 & MRR \\ \toprule
$\lambda$MART ($10^3$) & 0.397 & 0.434 & 0.481 & 0.500 \\
$\lambda$MART ($10^4$) & \textbf{0.410} & \textbf{0.441} & \textbf{0.492} & \textbf{0.506} \\ \midrule
PA-RiskRanker & \textbf{0.408} & \textbf{0.445} & \textbf{0.494} & \textbf{0.508} \\
\textit{w/o PA-BCE} & 0.327 & 0.362 & 0.412 & 0.431 \\
\textit{w/o Self-Cross Trader Attention \& PA-BCE (Rankformer)} & 0.218 & 0.258 & 0.307 & 0.312 \\
\bottomrule
\end{tabular}%
}
\caption{Comparison of LambdaMART and Rankformer variants on NDCG and MRR metrics, calculated using the binary risky label.
The values represent the mean performance across three separate runs. 
PA-RiskRanker without Self-Cross Trader Attention and PA-BCE is equivalent to Rankformer.}
\label{tab:ranking_result}
\vspace{-0.5cm}
\end{table}

Table~\ref{tab:ranking_result} reveals a significant advancement in the PA-RiskRanker model through the incorporation of Self-Cross Trader Attention and the PA-BCE loss. 
PA-RiskRanker surpasses LambdaMART with $10^3$ trees and achieves comparable results to LambdaMART with $10^4$ trees on ranking metrics. 
It is noteworthy that while Rankformer (PA-RiskRanker without Self-Cross Trader Attention and PA-BCE) demonstrated performance comparable to LambdaMART on standard information retrieval datasets \citep{rankformer}, it faces challenges when directly applied to domain-specific datasets with differing characteristics.
This aligns with our objective, which is not to significantly outperform LambdaMART on general-purpose ranking metrics, but rather to close the gap between naive deep learning–based rankers such as Rankformer and strong tree-based baselines, while paving the way for improved performance on domain-specific financial objectives such as P\&L. 
This follows surveys of trading exchange risk managers and has been highlighted by previous literature \citep{dnn}. 
We evaluate such metrics separately in Section~\ref{sec:pred_with_prior}.

More specifically, the incorporation of embedding layers can enhance the model's robustness, as evidenced by an increase in NDCG@10 from 0.307 to 0.412 and MRR from 0.312 to 0.431 with just the embedding addition. 
This highlights the necessity of feature embeddings for ranking on tabular data. 
On the other hand, the customisation of the loss function to address domain-specific imbalanced datasets, as seen with the PA-BCE loss, has led to a substantial performance boost. 
In particular, integrating PA-BCE also resulted in a marked improvement in NDCG@10, increasing from 0.412 to 0.494, and in MRR, rising from 0.431 to 0.508. 
Such gains were also seen with the self-cross trader attention, though its benefit is maximised when combined with PA-BCE. These results highlight that both embedding layers and profit-aware loss are essential for risk ranking. 

\subsection{Prediction Performance}
\label{sec:pred_with_prior}

While our initial findings show improved performance in general ranking evaluation metrics like NDCG and MRR, these may not directly align with our specific task requirements. 
Therefore, we also aim to assess performance using traditional classification and anomaly detection metrics.
We showcase two evaluation methods for simulating two scenarios whether or not the market makers have the prior knowledge of the proportion of risky traders (1\%).
While more sophisticated methods exist for injecting prior knowledge into classification models, it's important to clarify that our key focus is on the comparison of benchmark performances, with the incorporation of prior knowledge serving merely as a contextual aspect of our model evaluation rather than a focal point of our research.

\paragraph{With Prior.}
To facilitate a fair comparison between ranking models and other benchmark models, we adopt a scenario with the prior knowledge that 1\% of traders are risky. 
Practically, for ranking models, we predict scores within each ranking group, concatenate these groups, and sort all candidates by their scores. 
The top 1\% traders with the highest scores are then labelled as risky. 
For classification and anomaly detection models, each trader receives a predictive score. These scores are similarly sorted, marking the top 1\% as risky.

\paragraph{Without Prior.}
Additionally, we conduct evaluations without prior knowledge. Here, instead of sorting by ranking scores, we apply a \textit{Sigmoid} function to these scores. Considering that ranking scores may not use 0.5 as a standard binary prediction threshold, we perform a grid search over thresholds in the range of \([0.1, 0.2, ..., 0.9]\) on the validation set, selecting the one yielding the best performance. This optimal threshold search is equally applied to classification and anomaly detection models to maintain evaluative fairness.

\begin{table}[h!]
\centering
\resizebox{0.95\textwidth}{!}{%
\begin{tabular}{llrrrrrr}
\toprule
Model & Type & F1 & P\&L & AUC & Precision & Sensitivity & Specificity \\ \midrule
\multicolumn{8}{c}{With Prior} \\ \midrule
DeepSAD & Anomaly Detection & 0.527 & 123.730 & 0.611 & 0.063 & 0.063 & 0.991 \\
SLAD & Anomaly Detection & 0.500 & 121.978 & 0.520 & 0.010 & 0.010 & 0.990 \\
DIF & Anomaly Detection & 0.498 & 124.428 & 0.486 & 0.007 & 0.007 & 0.990 \\
FeaWAD & Anomaly Detection & 0.498 & 121.483 & 0.481 & 0.005 & 0.005 & 0.990 \\ \midrule
FT-Transformer & Classification & 0.611 & 129.552 & 0.859 & 0.229 & 0.231 & 0.992 \\
XGB Classifier ($10^4$) & Classification & 0.631 & 128.543 & 0.883 & 0.268 & 0.270 & 0.993 \\
LGBM Classifier ($10^4$) & Classification & 0.607 & 126.580 & 0.863 & 0.221 & 0.224 & 0.992 \\
Random Forest & Classification & 0.585 & 125.726 & 0.831 & 0.177 & 0.178 & 0.992 \\ \midrule
$\lambda$MART ($10^4$) & Ranking & \textbf{0.649} & 127.889 & 0.884 & 0.304 & 0.307 & 0.993 \\
PA-$\lambda$MART ($10^4$) & Ranking & 0.614 & 123.315 & 0.887 & 0.234 & 0.236 & 0.992 \\
SOUR & Ranking & 0.501 & 122.248 & 0.364 & 0.012 & 0.012 & 0.990 \\
Rankformer & Ranking & 0.564 & 128.963 & 0.827 & 0.137 & 0.138 & 0.991 \\
LambdaLoss & Ranking & 0.565 & 123.947 & 0.831 & 0.138 & 0.139 & 0.991 \\ 
PA-RiskRanker & Ranking & 0.643 & \textbf{142.517} & \textbf{0.898} & 0.292 & 0.295 & 0.993 \\ \midrule
\multicolumn{8}{c}{Without Prior} \\ \midrule
DeepSAD & Anomaly Detection & 0.498 & 123.249 & 0.486 & 0.014 & 0.001 & 1.000 \\
SLAD & Anomaly Detection & 0.494 & 106.293 & 0.611 & 0.024 & 0.244 & 0.901 \\
DIF & Anomaly Detection & 0.481 & 113.435 & 0.520 & 0.009 & 0.079 & 0.908 \\
FeaWAD & Anomaly Detection & 0.478 & 114.132 & 0.481 & 0.007 & 0.068 & 0.901 \\ \midrule
FT-Transformer & Classification & 0.611 & 129.714 & 0.859 & 0.231 & 0.230 & 0.992 \\
XGB Classifier ($10^4$) & Classification & 0.628 & 129.412 & 0.883 & 0.366 & 0.203 & 0.997 \\
LGBM Classifier ($10^4$) & Classification & 0.590 & 126.056 & 0.863 & 0.282 & 0.140 & 0.997 \\
Random Forest & Classification & 0.523 & 124.624 & 0.831 & 0.546 & 0.027 & 1.000 \\ \midrule
$\lambda$MART ($10^4$) & Ranking & \textbf{0.648} & 127.001 & 0.884 & 0.348 & 0.266 & 0.995 \\
PA-$\lambda$MART ($10^4$) & Ranking & 0.595 & 123.046 & 0.887 & 0.138 & 0.401 & 0.975 \\
SOUR & Ranking & 0.483 & 113.368 & 0.364 & 0.007 & 0.059 & 0.918 \\
Rankformer & Ranking & 0.528 & 127.400 & 0.827 & 0.165 & 0.039 & 0.998 \\
LambdaLoss & Ranking & 0.567 & 124.936 & 0.831 & 0.119 & 0.189 & 0.986 \\
PA-RiskRanker & Ranking & 0.633 & \textbf{140.648} & \textbf{0.894} & 0.227 & 0.353 & 0.988 \\ \bottomrule
\end{tabular}%
}
\caption{Detailed performance breakdown of benchmark models across anomaly detection, classification, and ranking, in contexts with and without prior knowledge, focusing on F1, P\&L, AUC, and other metrics. The values represent the mean performance across three separate runs.}
\label{tab:predict_results}
\vspace{-0.3cm}
\end{table}

The comparative results presented in Table~\ref{tab:predict_results} indicate a clear trend: classification and ranking methods generally outperform anomaly detection algorithms on our dataset, particularly in key metrics such as F1 score, P\&L, and AUC score.
This trend is consistent across scenarios both with and without prior knowledge, underscoring the robustness of these methods. Despite the utility of anomaly detection in specific contexts, it falls short in effectiveness compared to classification and ranking for this particular task.

The best results are observed from ranking models. Our PA-RiskRanker model showcases a remarkable performance, particularly in terms of P\&L. Specifically, PA-RiskRanker yields a P\&L of £142.517 (per 20 trades) in the ``with prior'' scenario, and £140.648 in the ``without prior'' scenario. These values notably surpass those of FT-Transformer, the next best-performing model in terms of P\&L, which records £129.552 and £129.714 respectively in the same settings. This indicates a substantial improvement in P\&L by approximately 10.0\% and 8.4\% in each scenario.
As aforementioned, given the scale of the financial market \cite{https://doi.org/10.1111/j.0306-686X.2004.00552.x}, this improvement per 20 trades can rapidly accumulate across millions of transactions.

In terms of the F1 score, PA-RiskRanker also demonstrates robust performance, achieving scores of 0.643 and 0.633 in the ``with prior'' and ``without prior'' scenarios, respectively. These scores are competitive with LambdaMART, which scores slightly higher at 0.649 and 0.648, respectively. However, the difference in F1 scores is marginal compared to the substantial leap in P\&L.

Furthermore, the AUC scores of PA-RiskRanker, standing at 0.898 and 0.894 across the two scenarios, also indicate superior performance, especially when compared to other ranking models like the standard Rankformer. This reinforces the model's capability in discriminating between risky and non-risky traders effectively.

\begin{wrapfigure}{R}{0.5\textwidth}
    \centering
    \includegraphics[width=0.48\textwidth]{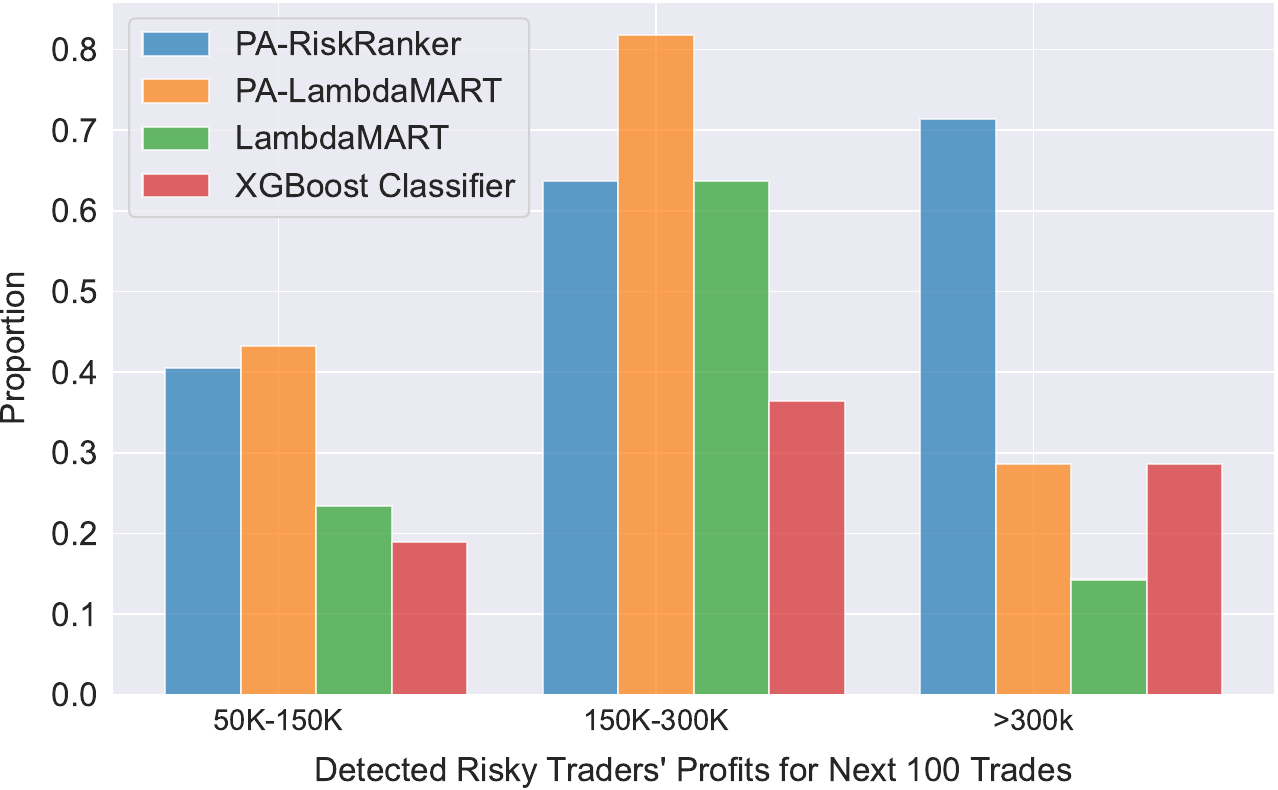}
    \caption{Proportion of correctly identified risky traders segmented by future profit ranges: £50K–£150K, £150K–£300K, and above £300K, across selected benchmark models.}
    \label{fig:anomly-profit-distribution}
\end{wrapfigure}

On the other hand, we observed that PA-$\lambda$MART does not outperform the original $\lambda$MART using the \textit{lambdarank} objective, suggesting that our PA-BCE loss alone is not fully compatible with the $\lambda$MART architecture. 
Specifically, in the ``without prior'' setting, PA-$\lambda$MART exhibits increased sensitivity but a drop in precision, leading to over-identification of potential risks. 
To probe this, we analysed traders whose future profits exceed £50K over the next 100 trades, using PA-RiskRanker, PA-$\lambda$MART, LambdaMART, and XGBoost Classifier. 
Notably, Figure~\ref{fig:anomly-profit-distribution} shows that PA-RiskRanker accurately identified over 71.4\% of traders with profits exceeding £300K, whereas LambdaMART and XGBoost identified only 28.6\%. 
This demonstrates the unique strength of PA-RiskRanker in capturing extreme-risk cases. 
While PA-$\lambda$MART detects more risky traders overall, its balance between sensitivity and precision is suboptimal without self-cross-trader attention. 
These results further highlight that it is the combination of the PA-BCE loss with our self-cross-trader attention mechanism that delivers superior detection of high-risk traders, with both components playing crucial, complementary roles.

\subsection{Two-step Evaluation}

In real-world settings, market makers favour interpretable models for predictions to facilitate analysis of feature importance and derive professional insights from the result \cite{CHEN2024357,DUMITRESCU20221178}. Accordingly, we have designed a two-step evaluation process. This approach assesses the extent to which initial ranking or predictive scores from less interpretable models can enhance overall predictive performance when integrated with more interpretable classification models.

Specifically, in the first step, a ranking algorithm or machine learning algorithm is used to generate ranking or predicting scores for the instances.
Subsequently, these generated scores are integrated into the dataset as an additional feature.
In the second step, interpretable classification models are trained on the updated data to provide binary predictions. 

For a comprehensive evaluation, we employ three machine learning methods, Random Forest, LGBM Classifier, and XGB Classifier, as second-step models in our two-step evaluation process.
To highlight the improvements brought to smaller models, we utilise LGBM and XGB classifiers with $10^3$ trees in the second step. This choice, favouring computational efficiency, is adequately representative for relative comparisons among models.
Each first-step benchmark model is paired with these classifiers, and we calculate their average performance to assess their collective efficacy thoroughly. To ensure fairness in our evaluations, we also conduct a grid search for the optimal threshold in the second-step classification models, aligning with our approach in previous experiments.

\begin{wraptable}{R}{0.5\textwidth}
\centering
\resizebox{0.42\textwidth}{!}{%
\begin{tabular}{lrrr}
\toprule
First Step Model & F1 & P\&L & AUC \\ \toprule
Single Step Avg. & 0.551 & 125.479 & 0.865 \\ \midrule
Random Forest & \cellcolor[HTML]{FFDFDE}-0.001 & \cellcolor[HTML]{D6FFD6}+0.160 & \cellcolor[HTML]{FFDFDE}-0.021 \\
XGB Classifier & \cellcolor[HTML]{D6FFD6}+0.024 & \cellcolor[HTML]{FFDFDE}-0.334 & \cellcolor[HTML]{FFDFDE}-0.166 \\
LGBM Classifier & \cellcolor[HTML]{D6FFD6}+0.023 & \cellcolor[HTML]{FFDFDE}-3.415 & \cellcolor[HTML]{FFDFDE}-0.120 \\
FT-Transformer & \cellcolor[HTML]{D6FFD6}+0.027 & \cellcolor[HTML]{D6FFD6}+2.611 & \cellcolor[HTML]{FFDFDE}-0.018 \\
DeepSAD & \cellcolor[HTML]{D6FFD6}+0.066 & \cellcolor[HTML]{D6FFD6}+1.227 & \cellcolor[HTML]{D6FFD6}+0.001 \\
DIF & \cellcolor[HTML]{FFDFDE}-0.011 & \cellcolor[HTML]{FFDFDE}-0.877 & \cellcolor[HTML]{FFDFDE}-0.007 \\
FeaWAD & \cellcolor[HTML]{D6FFD6}\textbf{+0.073} & \cellcolor[HTML]{D6FFD6}+2.519 & \cellcolor[HTML]{D6FFD6}\textbf{+0.004} \\
SLAD & \cellcolor[HTML]{D6FFD6}+0.061 & \cellcolor[HTML]{D6FFD6}+2.223 & \cellcolor[HTML]{FFDFDE}-0.003 \\
$\lambda$MART & \cellcolor[HTML]{D6FFD6}+0.040 & \cellcolor[HTML]{D6FFD6}+1.756 & \cellcolor[HTML]{FFDFDE}-0.164 \\
PA-$\lambda$MART & \cellcolor[HTML]{D6FFD6} +0.051 & \cellcolor[HTML]{D6FFD6} +1.740 & \cellcolor[HTML]{FFDFDE} -0.065 \\
SOUR & \cellcolor[HTML]{D6FFD6} +0.064 & \cellcolor[HTML]{D6FFD6} +3.491 & \cellcolor[HTML]{FFDFDE} -0.004 \\
Rankformer & \cellcolor[HTML]{D6FFD6}+0.053 & \cellcolor[HTML]{D6FFD6}+0.817 & \cellcolor[HTML]{FFDFDE}-0.006 \\ 
LambdaLoss & \cellcolor[HTML]{D6FFD6} +0.055 & \cellcolor[HTML]{FFDFDE} -0.821 & \cellcolor[HTML]{FFDFDE} -0.005 \\
\midrule
PA-RiskRanker & \cellcolor[HTML]{D6FFD6}+0.061 & \cellcolor[HTML]{D6FFD6}\textbf{+5.154} & \cellcolor[HTML]{FFDFDE}-0.005 \\
\bottomrule
\end{tabular}%
}
\caption{Performance comparison of various models in the two-step evaluation process, highlighting average changes in F1, P\&L, and AUC metrics when combined with second-step classifiers.}
\label{tab:twostep_results}
\end{wraptable}

Table~\ref{tab:twostep_results} summarises the key metrics such as F1 score, P\&L, and AUC for conciseness. The results indicate that anomaly detection models, particularly FeaWAD and DeepSAD, generally enhance the performance of second-step classifiers. Notably, FeaWAD leads with the highest increase in F1 score (+0.073) and a significant rise in P\&L (+2.519), while DeepSAD shows a substantial improvement in F1 score (+0.066) and a moderate increase in P\&L (+1.227). This trend underscores the value of integrating anomaly detection models into the two-step evaluation framework.

However, the most remarkable outcome is observed with the PA-RiskRanker model. 
It achieves the most substantial enhancement in terms of P\&L, with an impressive increase of +5.154, far exceeding other models. 
This significant improvement in P\&L highlights the PA-RiskRanker's capability to effectively identify profitable trading opportunities. 
Its F1 score also sees a notable rise of +0.061, although this is slightly lower than FeaWAD.

In addition, we examine the feature importance within the classifiers to provide a practical example of interpretability analysis. 
Appendix~\ref{appendix:feature-importance} showcases this aspect by demonstrating the significance of each feature in the classifiers’ decision-making processes. 
While this serves as a simple example, more sophisticated and domain-specific techniques could be applied for a deeper analysis. 
Importantly, our goal is to show that the improvements achieved by PA-RiskRanker are not only beneficial on their own but can also enhance the performance of other interpretable models or be integrated into ensemble learning frameworks.

This experiment confirms that our model significantly enhances the performance of interpretable classifiers, validating its practical applicability and effectiveness. 
This suitability is especially relevant in contexts where understanding and interpreting model predictions are as crucial as their performance.

\section{Conclusion}

Our research introduces a novel approach to detecting risky traders by advocating for a ranking-based methodology over traditional classification or anomaly detection techniques. 
The key innovation in our study is the introduction of the Profit-Aware Binary Cross Entropy (PA-BCE) loss function, integrated with a Self-Cross Trader Attention pipeline in the PA-RiskRanker model. 
This combination marks a significant advancement in the field, as it tailors the model to address the specific complexities and nuances of financial data.

Empirical evidence from our comprehensive experiments demonstrates the strength of our approach, particularly in optimising Profit and Loss (P\&L), a vital metric in the financial trading sector. 
The PA-RiskRanker model not only excels in traditional ranking metrics when compared to existing deep learning-based ranking models, but also surpasses recent models in classification and anomaly detection in terms of predictive performance.

Furthermore, our two-step evaluation process demonstrates the model's capacity to improve the performance of more interpretable classifiers. 
This aspect is particularly important for practical financial decision-making, where understanding the model's reasoning is as crucial as its predictive accuracy.

Overall, this study paves the way for the broader application of LEarning-TO-Rank (LETOR) algorithms in financial risk management, offering a new perspective on modelling and predicting individual behaviours in the financial context. 

Looking forward, we see considerable potential in extending our framework to related domains such as credit scoring, anti-money laundering, and fraud detection. 
Each of these application areas poses unique data structures, risk definitions, and regulatory requirements that would necessitate careful adaptation of our approach. 
For example, in credit scoring, the profit-aware metric may require redefinition to reflect long-term creditworthiness or loss given default \cite{Liao_Wang_Xue_Lei_Han_Lu_2022}, while anti-money laundering scenarios could benefit from modelling complex transaction networks or temporal sequences of activity \cite{Egressy_von_Niederhäusern_Blanuša_Altman_Wattenhofer_Atasu_2024,10.1145/3696410.3714576}.

Importantly, many real-world datasets in these domains include rich textual information \cite{han-etal-2018-nextgen,ahbali-etal-2022-identifying}, such as transaction descriptions, free-text justifications, or customer communications, which are not fully leveraged by traditional structured data models. Future work could incorporate advanced natural language processing techniques or pretrained language models to extract additional signal from these textual fields, integrating them directly into the risk-aware ranking framework. Adapting the PA-BCE loss and the self-cross attention mechanism to seamlessly fuse structured and unstructured features represents a critical next step. We believe this methodological flexibility is essential for achieving robust, domain-specific performance and will enable our approach to deliver practical value across a broad spectrum of high-stakes financial applications.

\begin{acks} 
We express our gratitude to the reviewers for their insightful comments, which significantly enhanced this paper. The experiments were supported by the Edinburgh Compute and Data Facility (ECDF). 
\end{acks}

\bibliographystyle{ACM-Reference-Format}
\bibliography{sample-base}

\appendix

\section{Proofs of Propositions}

\input{proofs}

\section{The Necessity For Novel Weighted Objective}
\label{sec:preliminary}
In the preliminary experiment, we demonstrate that simply using weighted objective cannot effectively solve the data imbalance issue and highlight the necessary need of novel objective design.

The setup is straightforward: assigning greater weights to positive (risky trader) instances during the training. This method is applied to the prevalent ranking loss function for binary labels in recent LETOR methodologies, \textit{Binary Cross Entropy} (BCE) \cite{pang2020setrank, rankformer}. The modified weighted BCE loss function (W-BCE) is formulated in Equation~\eqref{eq:wbce-loss}.

\begin{equation}
\label{eq:wbce-loss}
\mathcal{L}_\text{W-BCE}(\mathbf{s}, \mathbf{t}) = \sum - \mathbf{w} \cdot [ \mathbf{t} \log \sigma(\mathbf{s}) + (1-\mathbf{t}) \log (1-\sigma (\mathbf{s})) ]
\end{equation}

Here, $\mathbf{s}$ is the score vector for traders, $\mathbf{t}$ is the target status (1 for risky, 0 for non-risky), $\sigma(\mathbf{s})$ is the sigmoid function, and $\mathbf{w}$ is the weight vector emphasising risky traders.




\begin{wraptable}{R}{0.6\textwidth}
\centering
\resizebox{0.6\textwidth}{!}{%
\begin{tabular}{lllrrrr}
\toprule
Loss & Train & Test & NDCG@3 & NDCG@5 & NDCG@10\\ \midrule
LogSoftmax & 50 & 100 & 0.214 & 0.253 & 0.311 \\
BCE & 50 & 100 & 0.142 & 0.168 & 0.217  \\
W-BCE & 50 & 100 & 0.141 & 0.165 & 0.202 \\
$\lambda$MART & 50 & 100 & \textbf{0.420} & \textbf{0.450} & \textbf{0.496} \\
LogSoftmax & 100 & 100 & 0.218 & 0.258 & 0.307 \\
BCE & 100 & 100 & 0.134 & 0.157 & 0.201  \\
W-BCE & 100 & 100 & 0.127 & 0.148 & 0.183 \\
$\lambda$MART & 100 & 100 & \textbf{0.420} & \textbf{0.450} & \textbf{0.496}  \\
\bottomrule
\end{tabular}
}
\caption{Preliminary results on weighted \textit{LogSoftmax} and \textit{BCE}.}
\label{tab:preli_results}
\vspace{-0.5cm}
\end{wraptable}

The evaluations here are conducted to compare the weighted BCE loss function  against the original BCE loss and \textit{LogSoftmax} which is another widely-used loss function.
The weights vector $\mathbf{w}$ is assigned a value of 2 for instances of risky traders. Further tests with weights of 5 and 10 were conducted, yielding consistent performance across these variations.

The backbone ranking algorithm we used is Rankformer \cite{rankformer} with original hyperparameters setup. We train the Rankformer model with different loss functions using group size of 50, 100 on training set and 100 on test set so that we ensure every group contains a risky trader during testing.
We also implement the LambdaMART model with $10^4$ trees to establish a strong baseline, acknowledging existing research that demonstrates transformer-based ranking models, such as Rankformer and SetRank, often struggle to exceed the performance benchmarks set by LambdaMART in ranking tasks \cite{pang2020setrank,rankformer}. 
This choice is informed by the current consensus in the field regarding the efficacy of GBDT based methods \cite{xgboost,NIPS2017_6449f44a,NEURIPS2022_0378c769}.

Based on the results presented in Table~\ref{tab:preli_results}, it is evident that the naive weighted loss functions do not outperform the original method.
Specifically, the implementation of W-BCE leads to a marginal decline in ranking performance, falling below our expectations. Conversely, the LambdaMART model significantly outperforms Rankformer with all the benchmark objectives across all metrics, thereby underscoring the superiority of GBDT-based ranking models. These findings highlight the need for further refinement in our approach and validate the exploration of our novel PA-BCE objective.

\section{Hyper-parameters of Benchmark Models}
\label{appendix:hyper-parameters}

Table~\ref{tab:hyper-parameters} outlines the hyper-parameters for benchmark models, encompassing anomaly detection, classification, and ranking methods.

\begin{table}[htbp]
\centering
\resizebox{\textwidth}{!}{%
\begin{tabular}{lll}
\toprule
Model & Type & Hyper-Parameter \\ \midrule
DeepSAD & Anomaly Detection & epochs: 200, batch size: 64, learning rate: 0.001, hidden states: {[}100, 50{]} \\
SLAD & Anomaly Detection & epochs: 200, batch size: 64, hidden states: {[}100, 50{]} \\
DIF & Anomaly Detection & epochs: 200, batch size: 128, hidden states: 256 \\
FeaWAD & Anomaly Detection & epochs: 200, batch size: 64, learning rate: 0.001, hidden states: {[}100, 50{]} \\
FT-Transformer & Classification & epochs: 500, batch size: 1024, learning rate: automatic, transformer layers: 6 \\
XGB Classifier & Classification & trees: $10^4$ \\
LGBM Classifier & Classification & trees: $10^4$ \\
Random Forest & Classification & max depth: 7 \\
LambdaMART & Ranking & objective: lambdarank, trees: $10^4$ \\
Rankformer & Ranking & epochs: 600, learning rate: 0.0001, feedforward: 128, attention heads: 8, transformer layers: 6 \\
PA-RiskRanker & Ranking & epochs: 200, learning rate: 0.0001, feedforward: 128, attention heads: 2, transformer layers: 6 \\ \bottomrule
\end{tabular}%
}
\caption{Overview of hyperparameters for benchmark models. Any parameters not listed are retained at their default values.}
\label{tab:hyper-parameters}
\vspace{-0.7cm}
\end{table}

\section{Reproducible Experiment on Extra Dataset}
\label{appendix:reproducible-exp}

\subsection{Credit Card Fraud Detection}

In this appendix section, we introduce our reproducible experiment on the widely-used credit card fraud detection dataset\footnote{\url{https://www.kaggle.com/datasets/mlg-ulb/creditcardfraud}} from Kaggle due to its relevance to our task of detecting risky traders.

\paragraph{Data Preprocessing}
Unlike trading data, this fraud detection dataset lacks information on ``future profit'', which we used to construct the $\mathbf{G}_{P\&L}$ matrix. 
Consequently, we treat the \texttt{Amount} column as a proxy for financial loss and use it to generate ranking labels, replacing the original fraudulent label.
To align the ``1\% risky trader'' proportion in our dataset, we simply sort the credit card data by its transactional amount in descending order and mark the top 1\% as positive.
After obtaining the labels, we then follow the ranking group allocation procedure specified in Section~\ref{sec:data-preprocess} to generate ranking groups that allow ranking benchmarks to be trained on.
The train-validation-test split is 70\%-10\%-20\% and we ensure the minority-to-majority class ratio of 1\% to 99\%.

\begin{table}[ht]
\centering
\resizebox{\columnwidth}{!}{%
\begin{tabular}{llrrrrrr}
\toprule
Model & Type & F1$\uparrow$ & Financial Loss$\downarrow$ & AUC$\uparrow$ & Precision$\uparrow$ & Sensitivity$\uparrow$ & Specificity$\uparrow$ \\ \midrule
\multicolumn{8}{c}{With Prior} \\ \midrule
DeepSAD & Anomaly Detection & 0.9767 & 52,499.95 & 0.9998 & 0.9539 & 0.9539 & 0.9995 \\
SLAD & Anomaly Detection & 0.5960 & 745,006.85 & 0.9687 & 0.2002 & 0.2002 & 0.9919 \\
DIF & Anomaly Detection & 0.5988 & 725,901.98 & 0.9652 & 0.2055 & 0.2055 & 0.9920 \\
FeaWAD & Anomaly Detection & 0.7072 & 486,894.08 & 0.9921 & 0.4203 & 0.4203 & 0.9941 \\ \midrule
FT-Transformer & Classification & 0.9593 & 88,705.93 & 0.9971 & 0.9195 & 0.9195 & 0.9992 \\
XGB Classifier ($10^4$) & Classification & 0.9755 & 58,149.26 & 0.9999 & 0.9516 & 0.9516 & 0.9995 \\
LGBM Classifier ($10^4$) & Classification & 0.9729 & 63,659.04 & 0.9999 & 0.9463 & 0.9463 & 0.9995 \\
Random Forest & Classification & 0.9682 & 74,770.60 & 0.9998 & 0.9370 & 0.9370 & 0.9994 \\ \midrule
$\lambda$MART ($10^4$) & Ranking & 0.9694 & 72,817.16 & 0.9999 & 0.9393 & 0.9393 & 0.9994 \\
PA-$\lambda$MART ($10^4$) & Ranking & 0.9667 & 83,895.95 & 0.999 & 0.9340 & 0.9340 & 0.9993 \\
SOUR & Ranking & 0.4949 & 1,070,615.77 & 0.0763 & 0.0000 & 0.0000 & 0.9899 \\
Rankformer & Ranking & 0.9820 & 43,821.78 & 0.9999 & 0.9644 & 0.9644 & 0.9996 \\
LambdaLoss & Ranking & 0.9540 & 104,886.57 & 0.9997 & 0.9089 & 0.9089 & 0.9991 \\
PA-RiskRanker & Ranking & \textbf{0.9870} & \textbf{31,368.39} & 0.9998 & 0.9743 & 0.9743 & 0.9998 \\ \midrule
\multicolumn{8}{c}{Without Prior} \\ \midrule
DeepSAD & Anomaly Detection & 0.5664 & 884,417.04 & 0.9999 & 0.0994 & 1.0000 & 0.9085 \\
SLAD & Anomaly Detection & 0.5787 & 786,233.91 & 0.9687 & 0.1134 & 0.8832 & 0.9244 \\
DIF & Anomaly Detection & 0.6175 & 498,281.50 & 0.9652 & 0.1620 & 0.5709 & 0.9702 \\
FeaWAD & Anomaly Detection & 0.6965 & 657,491.46 & 0.9921 & 0.3058 & 0.8307 & 0.9629 \\ \midrule
FT-Transformer & Classification & 0.9593 & 89,030.83 & 0.9971 & 0.9007 & 0.9393 & 0.9990 \\
XGB Classifier ($10^4$) & Classification & 0.9718 & 67,644.60 & 0.9999 & 0.9758 & 0.9148 & 0.9998 \\
LGBM Classifier ($10^4$) & Classification & 0.9722 & 66,277.72 & 0.9999 & 0.9682 & 0.9229 & 0.9997 \\
Random Forest & Classification & 0.9649 & 85,672.36 & 0.9998 & 0.9683 & 0.8961 & 0.9997 \\ \midrule
$\lambda$MART ($10^4$) & Ranking & 0.9699 & 72,533.09 & 0.9999 & 0.9565 & 0.9247 & 0.9996 \\
PA-$\lambda$MART ($10^4$) & Ranking & 0.6741 & 1,190,792.17 & 0.9823 & 0.2244 & 0.9901 & 0.9654 \\
SOUR & Ranking & 0.4975 & 1,053,450.59 & 0.0763 & 0.0000 & 0.0000 & 1.0000 \\
Rankformer & Ranking & 0.9811 & 45,342.84 & 0.9999 & 0.9580 & 0.9673 & 0.9996 \\
LambdaLoss & Ranking & 0.9533 & 104,037.59 & 0.9997 & 0.8883 & 0.9288 & 0.9988 \\
PA-RiskRanker & Ranking & \textbf{0.9856} & \textbf{35,332.36} & 0.9998 & 0.9708 & 0.9726 & 0.9997 \\ \bottomrule
\end{tabular}%
}
\caption{Detailed performance of benchmark models on the credit card fraud detection dataset. The values represent the mean performance across 3-fold cross validation.}
\label{tab:credit-card-results}
\vspace{-0.8cm}
\end{table}

\paragraph{Experimental Results}
In line with our main experiments, we evaluate all benchmark models using 3-fold cross-validation. 
The average performance results are summarised in Table~\ref{tab:credit-card-results}. 
Financial loss is computed under a simplified assumption: the model is penalised by the transaction amount for each misclassified case. 
While this assumption does not reflect real-world banking practices, where false positives typically do not incur direct financial penalties, it provides a more balanced evaluation. 
Penalising only false negatives would bias models toward over-sensitivity to achieve zero loss, which is undesirable. 
Thus, our setup offers a reasonable trade-off for assessing performance in financial terms.

It is important to note that the \texttt{Amount} column, which we use to derive ranking labels, is originally an input feature in the dataset.
As described in the dataset documentation, the remaining features are the result of dimensionality reduction techniques such as Principal Component Analysis (PCA), which may contain strong correlations with \texttt{Amount}. 
This likely contributes to the high AUC scores observed across most of the benchmark models.

Nevertheless, PA-RiskRanker stands out by achieving the highest F1 score and the lowest financial loss across both ``with prior'' and ``without prior'' setups. The reduction in financial loss is particularly significant, with margins ranging from approximately £10K to £1000K compared to other models, further validating the effectiveness of our method.

\subsection{Job Profit Prediction}

In this appendix section, we further assess the generalisability of our approach on a job profitability prediction dataset\footnote{\url{https://www.kaggle.com/datasets/ulrikthygepedersen/job-profitability}}. The dataset records individual jobs with features relevant to job execution, costs, and revenue generation. 
To prevent data leakage, we remove columns containing future information, including \texttt{Job\_Number}, \texttt{Jobs\_Subtotal}, \texttt{Labor}, \texttt{Jobs\_Total}, \texttt{Lead\_Generated\_From\_Source}, \texttt{Pricebook\_Price}, and \texttt{Jobs\_Gross\_Margin}.

Following the same protocol as the credit card fraud detection experiment, we assign positive labels to the top 1\% of jobs ranked by profitability and negatives to the rest. Ranking group allocation, train-validation-test splits (70\%-10\%-20\%), and class imbalance settings are kept consistent with Section~\ref{appendix:reproducible-exp}.
Financial loss is also calculated identically: for each misclassified job, we penalise the model by the missed profit. All models are evaluated using 3-fold cross-validation, with results reported in Table~\ref{tab:my-table}.

\begin{table}[htp]
\centering
\resizebox{\textwidth}{!}{%
\begin{tabular}{llllllll}
\toprule
Model & Type & \multicolumn{1}{r}{F1 $\uparrow$} & \multicolumn{1}{r}{Financial Loss $\downarrow$} & \multicolumn{1}{r}{AUC $\uparrow$} & \multicolumn{1}{r}{Precision $\uparrow$} & \multicolumn{1}{r}{Sensitivity $\uparrow$} & \multicolumn{1}{r}{Specificity $\uparrow$} \\ \midrule
\multicolumn{8}{c}{With Prior} \\ \midrule
DeepSAD & Anomaly Detection & 0.4953 & 171,425.37 & 0.4842 & 0.0000 & 0.0000 & 0.9899 \\
SLAD & Anomaly Detection & 0.6435 & 112,685.96 & 0.9777 & 0.2738 & 0.3205 & 0.9927 \\
DIF & Anomaly Detection & 0.5971 & 162,701.42 & 0.9717 & 0.1905 & 0.2167 & 0.9918 \\
FeaWAD & Anomaly Detection & 0.7724 & 77,227.59 & 0.9810 & 0.5119 & 0.5986 & 0.9951 \\ \midrule
FT-Transformer & Classification & 0.8285 & 72,724.86 & 0.9949 & 0.6190 & 0.7152 & 0.9962 \\
XGB Classifier ($10^4$) & Classification & 0.9247 & 28,348.70 & 0.9994 & 0.7976 & 0.9210 & 0.9980 \\
LGBM Classifier ($10^4$) & Classification & 0.9244 & 27,838.33 & 0.9995 & 0.7976 & 0.9195 & 0.9980 \\
Random Forest & Classification & 0.8980 & 37,132.58 & 0.9983 & 0.7500 & 0.8610 & 0.9975 \\ \midrule
$\lambda$MART ($10^4$) & Ranking & 0.9247 & 28,665.63 & 0.9994 & 0.7976 & 0.9210 & 0.9980 \\
PA-$\lambda$MART ($10^4$) & Ranking & 0.8799 & 51,606.02 & 0.9965 & 0.7143 & 0.8257 & 0.9971 \\
SOUR & Ranking & 0.4953 & 171,425.37 & 0.4842 & 0.0000 & 0.0000 & 0.9899 \\
Rankformer & Ranking & 0.8539 & 59,177.19 & 0.9954 & 0.6666 & 0.7690 & 0.9966 \\
LambdaLoss & Ranking & 0.8342 & 69,817.75 & 0.9950 & 0.6309 & 0.7252 & 0.9963 \\
PA-RiskRanker & Ranking & \textbf{0.9491} & \textbf{19,363.32} & 0.9996 & 0.8571 & 0.9506 & 0.9986 \\ \midrule
\multicolumn{8}{c}{Without Prior} \\ \midrule
DeepSAD & Anomaly Detection & 0.5271 & 348,874.44 & 0.8776 & 0.0582 & 0.7214 & 0.9011 \\
SLAD & Anomaly Detection & 0.5523 & 384,875.61 & 0.9777 & 0.0834 & 0.9614 & 0.9075 \\
DIF & Anomaly Detection & 0.6021 & 217,733.40 & 0.9717 & 0.2519 & 0.6543 & 0.9658 \\
FeaWAD & Anomaly Detection & 0.5538 & 547,800.86 & 0.9810 & 0.0840 & 0.8981 & 0.9132 \\ \midrule
FT-Transformer & Classification & 0.8448 & 64,114.73 & 0.9949 & 0.7693 & 0.6352 & 0.9983 \\
XGB Classifier ($10^4$) & Classification & 0.9263 & 27,530.39 & 0.9994 & 0.8215 & 0.8895 & 0.9983 \\
LGBM Classifier ($10^4$) & Classification & \textbf{0.9369} & \textbf{21,695.83} & 0.9995 & 0.9078 & 0.8476 & 0.9993 \\
Random Forest & Classification & 0.9165 & 30,440.85 & 0.9983 & 0.8452 & 0.8329 & 0.9987 \\ \midrule
$\lambda$MART ($10^4$) & Ranking & 0.9171 & 31,888.59 & 0.9994 & 0.8141 & 0.8648 & 0.9982 \\
PA-$\lambda$MART ($10^4$) & Ranking & 0.8933 & 39,272.97 & 0.9986 & 1.0000 & 0.6500 & 1.0000 \\
SOUR & Ranking & 0.4975 & 161,824.86 & 0.4842 & 0.0000 & 0.0000 & 0.9989 \\
Rankformer & Ranking & 0.8190 & 83,849.98 & 0.9954 & 0.6084 & 0.7662 & 0.9944 \\
LambdaLoss & Ranking & 0.8115 & 81,905.58 & 0.9950 & 0.6439 & 0.6690 & 0.9959 \\
PA-RiskRanker & Ranking & 0.9359 & 28,266.70 & 0.9995 & 0.8529 & 0.9057 & 0.9986 \\ \bottomrule
\end{tabular}%
}
\caption{Detailed performance of benchmark models on the job profit prediction dataset. The values represent the mean performance across 3-fold cross validation.}
\label{tab:my-table}
\vspace{-0.5cm}
\end{table}

Consistent with previous experiments, PA-RiskRanker achieves the highest F1 and the lowest financial loss in the ``with prior'' setting, which aligns with our ranking-based problem formulation. 
In the ``without prior'' setting, prediction models such as LGBM Classifier and XGB Classifier outperform most ranking methods, with PA-RiskRanker remaining highly competitive, ranking as the second-best overall. 
This suggests that, while job profitability prediction differs from our core risk-ranking task and may favour traditional classifiers under certain evaluation setup, our approach still demonstrates strong transferability, delivering best performance in ranking settings and remaining robust and comparable in more conventional classification setups.

\section{Feature Importance Analysis}
\label{appendix:feature-importance}

Figures~\ref{fig:feat_imp_xgb}, \ref{fig:feat_imp_lgbm}, and \ref{fig:feat_imp_rf} present the feature importance scores obtained from the second-step classifiers, XGBoost, LightGBM, and Random Forest—each trained with $10^3$ trees. 
Features are ordered in ascending importance. 
Higher scores reflect stronger influence on the classifier’s output.

These results serve as an illustrative example of how the two-step approach contributes to interpretability by quantifying the influence of both raw behavioural features and the intermediate predictive scores generated in the first step (denoted as fst\_step\_scores). 
Notably, the relative importance of fst\_step\_scores differs across models: it ranks low in XGBoost, but is assigned moderate importance in LightGBM and Random Forest. 
This variation underscores the potential of using predicted scores from black-box models as input features for more interpretable classifiers, thereby providing a pathway to transparency through surrogate analysis.

\begin{figure}[htbp]
    \centering
    \includegraphics[width=\textwidth]{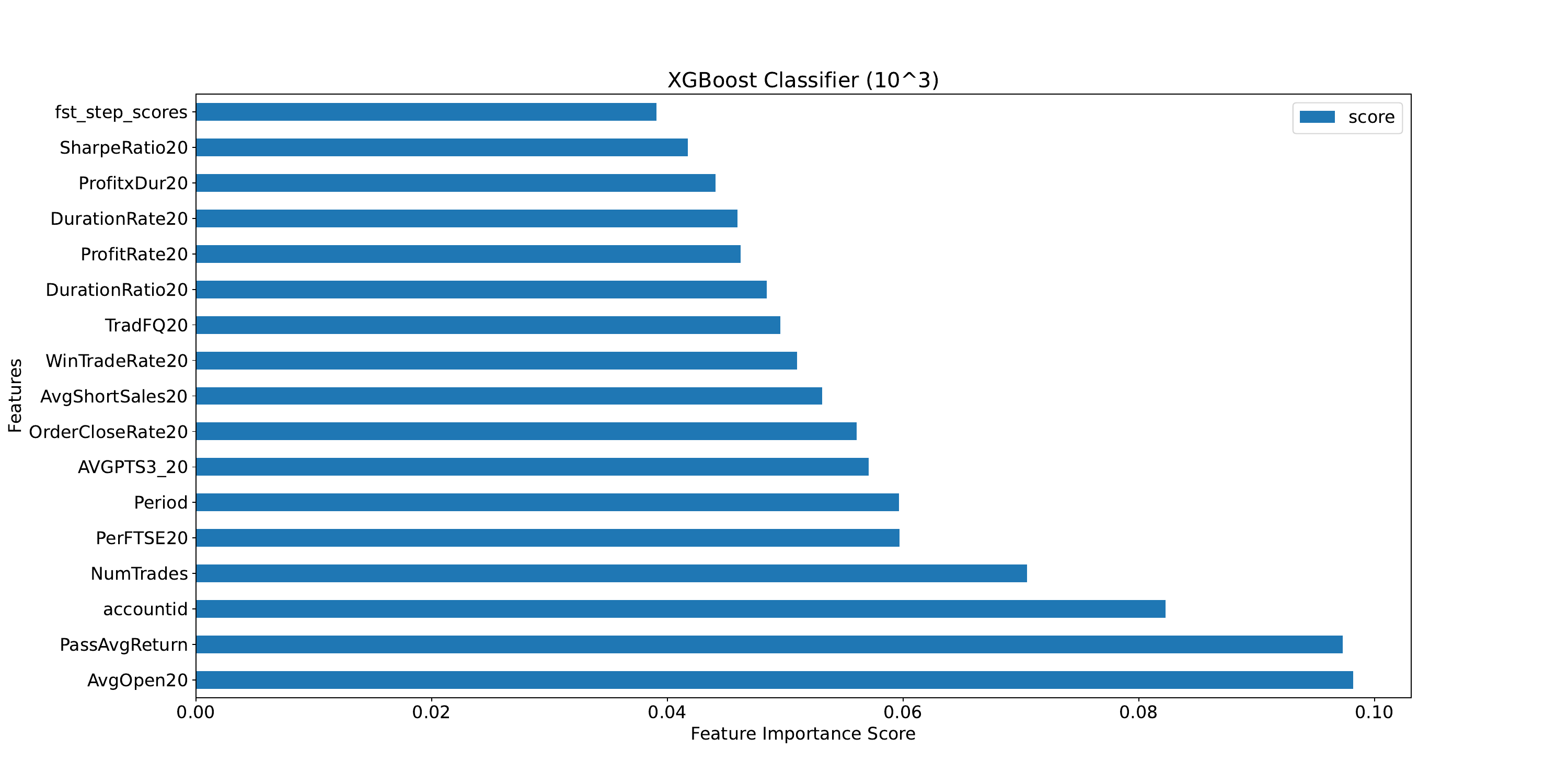}
    \caption{Feature importance of the second step XGBoost Classifier with $10^3$ trees. ``fst\_step\_scores'' stands for the predictive scores generated from the first step PA-RiskRanker.}
    \label{fig:feat_imp_xgb}
\end{figure}

\begin{figure}[htbp]
    \centering
    \includegraphics[width=\textwidth]{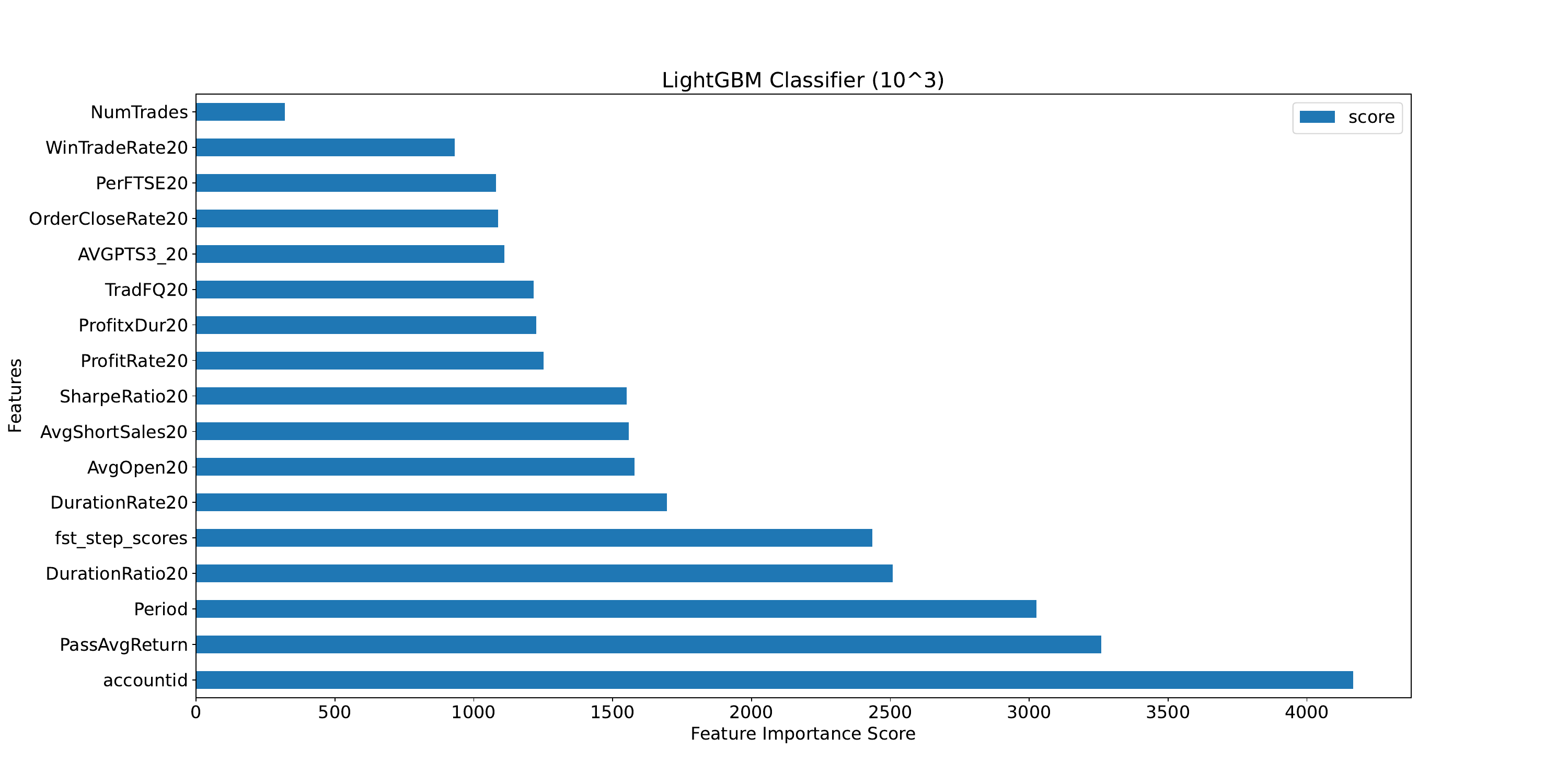}
    \caption{Feature importance of the second step LightGBM Classifier with $10^3$ trees. ``fst\_step\_scores'' stands for the predictive scores generated from the first step PA-RiskRanker.}
    \label{fig:feat_imp_lgbm}
\end{figure}

\begin{figure}[htbp]
    \centering
    \includegraphics[width=\textwidth]{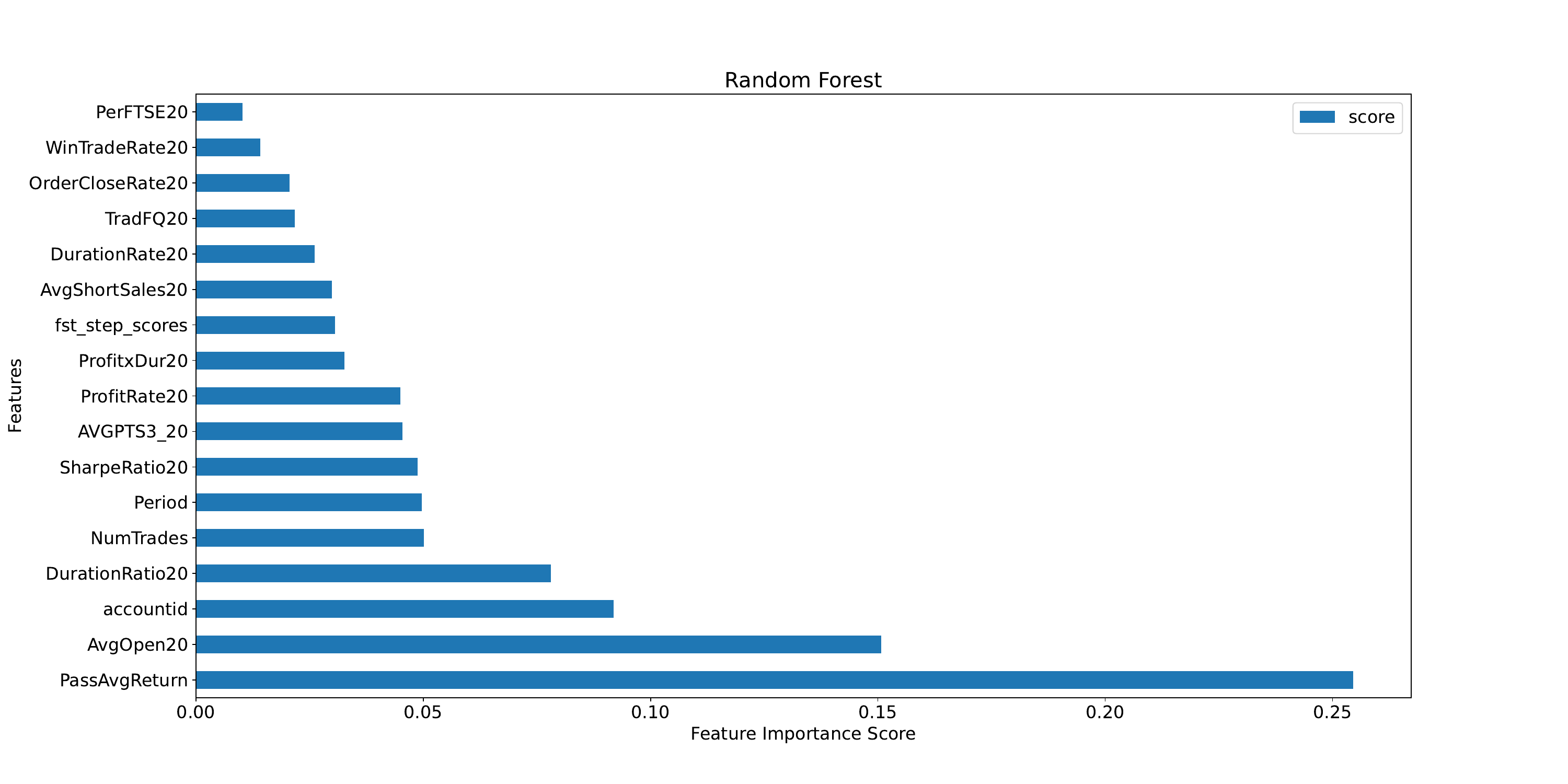}
    \caption{Feature importance of the second step Random Forest Classifier with $10^3$ trees. ``fst\_step\_scores'' stands for the predictive scores generated from the first step PA-RiskRanker.}
    \label{fig:feat_imp_rf}
\end{figure}

\end{document}

%% file: proofs.tex
\begin{theorem}
\label{proposition:balance}
Pairwise ranking approaches inherently produce a balanced class distribution by creating equal numbers of positive and negative labels for each pairwise comparison when each trader is associated with distinct future profits.
\end{theorem}
\renewcommand\qedsymbol{$\blacksquare$}
\begin{proof}
Let $D = \{ (x_i, p_i) \mid i = 1, \ldots, N \}$ be a dataset containing $N$ instances, where each instance $x_i$ has a corresponding future profit $p_i$, representing its risk. 
Given that for any two distinct instances $x_i$ and $x_j$ in $D$, we have $p_i \neq p_j$, meaning there are no instances with identical future profits. 
In pairwise ranking, for any two instances $(x_i, p_i)$ and $(x_j, p_j)$, we define a pairwise comparison label $y_{ij}$ as follows:

\[
y_{ij} = 
\begin{cases} 
1 & \text{if } p_i > p_j, \\
0 & \text{if } p_i < p_j.
\end{cases}
\]

We aim to show that this formulation inherently produces an equal number of positive and negative labels.

For each distinct pair $(i, j)$ where $i \neq j$, we generate two pairwise comparisons:
    \begin{itemize}
        \item $y_{ij}$, which is 1 if $p_i > p_j$.
        \item $y_{ji}$, which is 1 if $p_j > p_i$. This is equivalent to $p_{ji} = 1 - p_{ij}$.
    \end{itemize}
Therefore, for every positive label $y_{ij} = 1$, there exists a corresponding negative label $y_{ji} = 0$, ensuring a balance between positive and negative labels. 
\end{proof}

\begin{theorem}
\label{proposition:upper_triangular_equivalence}
Optimising the loss function \(\mathcal{L}_{\text{PA-BCE}}\) using only the upper triangular parts of the matrices \(\mathbf{G}_{\text{score}}\), \(\mathbf{G}_{P\&L}\), and \(\mathbf{T}\) is equivalent to optimising the loss over the entire matrices. 
Formally,
\begin{equation}
\label{eq:equivalent_loss}
    \mathcal{L}_{\text{PA-BCE}} = 2 \sum_{i<j} \mathbf{G}_{P\&L}(i,j) \cdot \text{BCE}\left( \sigma(s_i - s_j),\ \mathbf{T}(i,j) \right),
\end{equation}
where \(\sigma\) is the sigmoid function and \(\text{BCE}(\cdot, \cdot)\) is the binary cross-entropy loss.
\end{theorem}

\begin{proof}
The total loss over all pairs \((i, j)\) with \(i \neq j\) is defined as:
\begin{equation}
    \mathcal{L}_{\text{PA-BCE}} = \sum_{i \neq j} \ell_{ij} = \sum_{i \neq j} \mathbf{G}_{P\&L}(i,j) \cdot \text{BCE}\left( \sigma(s_i - s_j),\ \mathbf{T}(i,j) \right).
\end{equation}

Since the P\&L gap matrix \(\mathbf{G}_{P\&L}\) is symmetric, i.e., \(\mathbf{G}_{P\&L}(i,j) = \mathbf{G}_{P\&L}(j,i)\), and the target matrix satisfies \(\mathbf{T}(i,j) = 1 - \mathbf{T}(j,i)\) for all \(i \neq j\), we can group the loss terms for each pair \((i, j)\) as follows:
\begin{equation}
    \ell_{ij} + \ell_{ji} = \mathbf{G}_{P\&L}(i,j) \cdot \left[ \text{BCE}\left( \sigma(s_i - s_j),\ \mathbf{T}(i,j) \right) + \text{BCE}\left( \sigma(s_j - s_i),\ \mathbf{T}(j,i) \right) \right].
\end{equation}

Utilizing the property \(\sigma(s_j - s_i) = 1 - \sigma(s_i - s_j)\) and \(\mathbf{T}(j,i) = 1 - \mathbf{T}(i,j)\), we can express the combined loss as:
\begin{equation}
\label{eq:combined_loss}
    \ell_{ij} + \ell_{ji} = \mathbf{G}_{P\&L}(i,j) \cdot \left[ \text{BCE}\left( \sigma(s_i - s_j),\ \mathbf{T}(i,j) \right) + \text{BCE}\left( 1 - \sigma(s_i - s_j),\ 1 - \mathbf{T}(i,j) \right) \right].
\end{equation}

Recall that the binary cross-entropy loss is defined as:
\begin{equation}
    \text{BCE}(p, t) = -[t \log p + (1 - t) \log (1 - p)].
\end{equation}

Substituting this into Equation~\eqref{eq:combined_loss}, we have:
\begin{align}
    \ell_{ij} + \ell_{ji} &= \mathbf{G}_{P\&L}(i,j) \cdot \Big( -\left[ \mathbf{T}(i,j) \log \sigma(s_i - s_j) + (1 - \mathbf{T}(i,j)) \log (1 - \sigma(s_i - s_j)) \right] \nonumber \\
    &\quad\quad\quad\quad\quad\quad\ - \left[ (1 - \mathbf{T}(i,j)) \log (1 - \sigma(s_i - s_j)) + \mathbf{T}(i,j) \log \sigma(s_i - s_j) \right] \Big).
\end{align}

Simplifying the expression inside the brackets:
\begin{align}
    \ell_{ij} + \ell_{ji} &= \mathbf{G}_{P\&L}(i,j) \cdot \Big( -2 \left[ \mathbf{T}(i,j) \log \sigma(s_i - s_j) + (1 - \mathbf{T}(i,j)) \log (1 - \sigma(s_i - s_j)) \right] \Big) \nonumber \\
    &= 2 \mathbf{G}_{P\&L}(i,j) \cdot \text{BCE}\left( \sigma(s_i - s_j),\ \mathbf{T}(i,j) \right).
\end{align}

Therefore, the total loss becomes:
\begin{equation}
    \mathcal{L}_{\text{PA-BCE}} = \sum_{i \neq j} \ell_{ij} = \sum_{i<j} \left( \ell_{ij} + \ell_{ji} \right) = 2 \sum_{i<j} \mathbf{G}_{P\&L}(i,j) \cdot \text{BCE}\left( \sigma(s_i - s_j),\ \mathbf{T}(i,j) \right).
\end{equation}

Since the factor of 2 is constant across all terms, it does not affect the optima.
Thus, minimising \(\mathcal{L}_{\text{PA-BCE}}\) over all pairs is equivalent to minimising it over the upper triangular part of the matrices, scaled by a constant factor. 
This proves that optimising using only the upper triangular parts is sufficient and equivalent.
\end{proof}